\documentclass[a4paper,UKenglish,cleveref, autoref, thm-restate]{lipics-v2021}
\usepackage[utf8]{inputenc}
\usepackage{amsthm}
\usepackage{mathtools}
\usepackage{amsfonts}
\usepackage{amsmath}
\usepackage{graphicx}
\usepackage{tikz}
\usepackage{tabularx}

\graphicspath{{./graphics/}}%helpful if your graphic files are in another directory

\title{Long paths make pattern-counting hard, and deep trees make it harder}

\author{Vít Jelínek}
{Computer Science Institute, Charles University, Prague, Czechia}
{jelinek@iuuk.mff.cuni.cz}{https://orcid.org/0000-0003-4831-4079}{Supported by project 18-19158S of 
  the Czech Science 
  Foundation.}

\author{Michal Opler}
{Computer Science Institute, Charles University, Prague, Czechia}
{opler@iuuk.mff.cuni.cz}{https://orcid.org/0000-0002-4389-5807}{Supported by project 21-32817S of 
  the Czech Science 
  Foundation and by project SVV–2020–260578.}

\author{Jakub Pekárek}
{Department of Applied Mathematics, Charles University, Prague, Czechia}
{pekarej@kam.mff.cuni.cz}{https://orcid.org/0000-0002-5409-3930}{}

\authorrunning{V. Jelínek, M. Opler and J. Pekárek} %TODO mandatory. First: Use abbreviated 
\Copyright{Vít Jelínek, Michal Opler and Jakub Pekárek} %TODO mandatory, please use full first 


\category{} %optional, e.g. invited paper

\relatedversion{} %optional, e.g. full version hosted on arXiv, HAL, or other respository/website
\nolinenumbers %uncomment to disable line numbering

\hideLIPIcs  %uncomment to remove references to LIPIcs series (logo, DOI, ...), e.g. when preparing a pre-final version to be uploaded to arXiv or another public repository

\EventEditors{Petr A. Golovach and Meirav Zehavi}
\EventNoEds{2}
\EventLongTitle{16th International Symposium on Parameterized and Exact Computation (IPEC 2021)}
\EventShortTitle{IPEC 2021}
\EventAcronym{IPEC}
\EventYear{2021}
\EventDate{September 8--10, 2021}
\EventLocation{Lisbon, Portugal}
\EventLogo{}
\SeriesVolume{214}
\ArticleNo{22}
\newif\ifdraft
\draftfalse
\ifdraft % only for debugging
    \usepackage{todonotes}
    \usepackage{lineno}
    \linenumbers
\fi % only for debugging

\newcommand{\Grid}{\text{Grid}}%\Grid(M)=M-griddable permutations
\newcommand{\Geom}{\text{Geom}}
\newcommand{\PPM}{\textsc{PPM} }
\newcommand{\PPMx}{\textsc{PPM}}

\newcommand{\SPPM}{\textsc{\#PPM} }
\newcommand{\SPPMx}{\textsc{\#PPM}}

\newcommand{\PPPM}[1]{\textsc{#1-Pattern PPM}}
\newcommand{\SCPPPM}[1]{\textsc{#1-Pattern SCPPM}}
\newcommand{\APPPM}[1]{\textsc{#1-Pattern APPM}}
\newcommand{\PSPPM}[1]{\textsc{#1-Pattern \#PPM}}

\newcommand{\PSI}{\textsc{PSI} }
\newcommand{\PSIx}{\textsc{PSI}}

\newcommand{\cC}{\mathcal{C}}
\newcommand{\cD}{\mathcal{D}}
\newcommand{\cF}{\mathcal{F}}
\newcommand{\cL}{\mathcal{L}}
\newcommand{\cP}{\mathcal{P}}
\newcommand{\cQ}{\mathcal{Q}}
\newcommand{\cT}{\mathcal{T}}

\ifdraft
\newcommand\vj[1]{\todo[color=green!60!white]{\emph{VJ: #1}}}
\newcommand\mo[1]{\todo[color=orange!60!white]{\emph{MO: #1}}}
\newcommand\jp[1]{\todo[color=red!60!white]{\emph{JP: #1}}}
\else
\newcommand\vj[1]{}
\newcommand\mo[1]{}
\newcommand\jp[1]{}
\def\todo{}{}
\fi

\newcommand{\cM}{\mathcal{M}}
\newcommand{\cN}{\mathcal{N}}

\DeclareMathOperator{\tw}{tw}%tw_C(n):=max tw of permutations of size n in C
\DeclareMathOperator{\Av}{Av}
\DeclareMathOperator{\red}{red}
\DeclareMathOperator{\St}{St}

\newcommand{\Inc}{\mathrel{
		\begin{tikzpicture}[line cap=round, line join=round]
      \useasboundingbox (0.9ex, 0.9ex) rectangle (2.6ex,2.6ex);
			\draw (1ex, 1ex) rectangle (2.5ex, 2.5ex)
			(1ex, 1ex) -- (2.5ex, 2.5ex);
		\end{tikzpicture}
}}
\newcommand{\Dec}{\mathrel{
		\begin{tikzpicture}[line cap=round, line join=round]
      \useasboundingbox (0.9ex, 0.9ex) rectangle (2.6ex,2.6ex);
			\draw (1ex, 1ex) rectangle (2.5ex, 2.5ex)
			(2.5ex, 1ex) -- (1ex, 2.5ex);
		\end{tikzpicture}
}}

\makeatletter
\newcommand{\oset}[3][1.75ex]{%
  \mathrel{\mathop{#3}\limits^{
      \vbox to#1{\kern-2\ex@
        \hbox{$\tiny#2$}\vss}}}}
\makeatother

\newcommand{\oto}[1]{
\oset{\makebox[0pt]{\text{#1}}}{\to}
}

\newcommand{\idtile}[1]{\textsf{Id}(#1)}
\newcommand{\patterntile}[1]{\textsf{Pattern}(#1)}
\newcommand{\anchortile}{\textsf{Anchor}}
\newcommand{\assigntile}{\textsf{Assign}}
\newcommand{\branchtile}[2]{\textsf{Branch}(#1, #2)}
\newcommand{\testtile}[2]{\textsf{Test}(#1, #2)}
\newcommand{\mergetile}[2]{\textsf{Merge}(#1, #2)}
\newcommand{\msfid}{\mathsf{id}}
\newcommand{\msfa}{\mathsf{a}}
\newcommand{\msfb}{\mathsf{b}}
\newcommand{\msft}{\mathsf{t}}
\newcommand{\msfm}{\mathsf{m}}

\begin{document}

\maketitle

\begin{abstract}
We study the counting problem known as \#PPM, whose input is a pair of permutations $\pi$ and  $\tau$ 
(called \emph{pattern} and \emph{text}, respectively), and the task is to find the number of 
subsequences of $\tau$ that have the same relative order as~$\pi$. A simple brute-force approach 
solves \#PPM for a pattern of length $k$ and a text of length $n$ in time $O(n^{k+1})$, while 
Berendsohn, 
Kozma and Marx have recently shown that under the exponential time hypothesis (ETH), it cannot be 
solved in time $f(k) n^{o(k/\log k)}$ for any function~$f$. In this paper, we consider the 
restriction of \#PPM, known as \PSPPM{$\cC$}, where the pattern $\pi$ must belong to a hereditary 
permutation class~$\cC$. Our goal is to identify the structural properties of $\cC$ that determine 
the complexity of  \PSPPM{$\cC$}.

We focus on two such structural properties, known as the \emph{long path property} (LPP) and the 
\emph{deep tree property} (DTP). Assuming ETH, we obtain these results:
\begin{enumerate}
\item If $\cC$ has the LPP, then \PSPPM{$\cC$} cannot be solved in time 
$f(k)n^{o(\sqrt{k})}$ for any function $f$, and
\item if $\cC$ has the DTP, then \PSPPM{$\cC$} cannot be solved in time 
$f(k)n^{o(k/\log^2 k)}$ for any function~$f$.
\end{enumerate}
Furthermore, when $\cC$ is one of the so-called monotone grid classes, we show that if $\cC$ has the 
LPP but not the DTP, then \PSPPM{$\cC$} can be solved in time $f(k)n^{O(\sqrt k)}$. In particular, 
the lower bounds above are tight up to the polylog terms in the exponents.
%TODO: the claim above could perhaps be strenghtened?
\end{abstract}
%Requirment of IPEC is to have only the front matter on the first page.
\vfill
\pagebreak
\section{Introduction}
One of the most frequently studied algorithmic problems related to permutations is known as 
\textsc{Permutation pattern matching} (or \PPMx). The input of \PPM is a pair of permutations $\tau$ 
(the `text') of length $n$ and $\pi$ (the `pattern') of length $k$, and the goal is to determine 
whether $\tau$ contains $\pi$ as a subpermutation (see Section~\ref{sec-prelims} for formal 
definitions).

In full generality, \PPM is NP-complete, as shown by Bose et al.~\cite{BBL}. Thus most research 
into \PPM focuses either on improved exact algorithms, or on identifying special types of inputs 
for which the \PPM can be solved in polynomial time, or at least in subexponential time. Note 
that a direct brute-force approach solves \PPM in time $O(n^{k+1})$.

A particularly fruitful technique to solving \PPM has been proposed by Ahal and 
Rabinovich~\cite{AR08_subpattern}, who showed that \PPM can be solved in time $n^{O(\tw(\pi))}$, 
where $\tw(\pi)$ denotes the tree-width of the so-called \emph{incidence graph} of the 
pattern~$\pi$. The bound was subsequently tightened to $n^{\tw(\pi)+1}$ by Berendsohn, Kozma and 
Marx~\cite{Berendsohn19}, who have used it to show that \PPM can be solved in time $n^{k/4+o(k)}$. 

Another approach to \PPMx, due to Guillemot and Marx~\cite{Guillemot2014} (with a slight 
improvement by Fox~\cite{Fox}) shows that the problem can be solved in time $n\cdot 2^{O(k^2)}$, 
implying that the problem is fixed-parameter tractable with parameter~$k$.

Closely related to \PPM is its counting version \SPPMx, whose goal is to compute the number of 
occurrences of the pattern $\pi$ in the text~$\tau$. Berendsohn et al.~\cite{Berendsohn19} show 
that their bounds of $O(n^{\tw(\pi)+1})$ and $n^{k/4+o(k)}$ for \PPM also apply to solving \SPPMx. 
In contrast, the FPT result for \PPM by Guillemot and Marx~\cite{Guillemot2014} likely does not 
extend to \SPPMx, since Berendsohn et al.~\cite{Berendsohn19} show that, under the exponential time 
hypothesis (ETH), \SPPM cannot be solved in time $f(k)n^{o(k/\log k)}$, for any function~$f$.

Given that both \PPM and \SPPM are hard in general, it is natural to consider their complexity on 
restricted inputs. A common approach is to fix a hereditary class $\cC$ of permutations, and study 
the restriction of \PPM or \SPPM to inputs where the pattern $\pi$ belongs to~$\cC$. Such 
restriction is known as \PPPM{$\cC$} and \PSPPM{$\cC$}, respectively. It follows from the results 
of Ahal and Rabinovich~\cite{AR08_subpattern} and Berendsohn et al.~\cite{Berendsohn19}, that the 
restricted problems are polynomial whenever the function $\tw(\pi)$ is bounded on the class~$\cC$. 
This idea is the basis for previous results establishing sharp thresholds between polynomial and 
NP-hard cases of \PPPM{$\cC$}~\cite{JeKy,Jelinek2020}. In fact, in all the known cases when 
\PPPM{$\cC$} and \PSPPM{$\cC$} are polynomial, the class $\cC$ has bounded tree-width.

While distinguishing the polynomial cases of \PPPM{$\cC$} from the NP-hard ones is obviously the 
main focus of research, it is also of interest to distinguish subexponential cases from those cases 
which (under suitable complexity assumptions, such as the ETH) require exponential or 
near-exponential time. Here again, the tree-width plays a key role. It is convenient to associate to 
a 
class $\cC$ its \emph{tree-width growth function}
\[
\tw_\cC(k)=\max\{\tw(\pi);\; \pi\in\cC\land |\pi|=k\}.
\]
Indeed, Berendsohn et 
al.~\cite{Berendsohn19}, extending previous results by Guillemot and Vialette~\cite{GV09_321}, have 
shown that when $\cC$ is the class of 2-monotone permutations (i.e., the permutations merged from 
two monotone sequences), then $\tw_\cC(k)=O(\sqrt{k})$, and consequently \PSPPM{$\cC$} can be solved 
in the subexponential time $n^{O(\sqrt{k})}$. They show, however, that for the class of 3-monotone 
permutations, the tree-width growth is of order $\Omega(k/\log k)$. Later 
Berendsohn~\cite[Theorem 4.1]{BerendsohnMs} showed that for the class $\cC=\Av(654321)$, consisting 
of permutations that can be merged from 5 increasing subsequences, \PSPPM{$\cC$} cannot be solved in 
time $f(k)n^{o(k/\log^4 k)}$ for any function $f$, unless ETH fails.

In the context of \PPPM{$\cC$} and \PSPPM{$\cC$}, most of the research focuses on the cases when $\cC$ is a 
\emph{principal class}, i.e., the class $\Av(\sigma)$ of all the permutations that avoid a single forbidden pattern~$\sigma$. 
Unfortunately, principal classes seldom admit a suitable structural characterisation of their elements, and even in 
those cases where such characterisations exist, they are very different from one class to another. This makes it hard to
obtain general results that apply uniformly to a large set of principal classes. 

To sidestep this issue, we mostly avoid dealing with individual principal classes directly, and instead 
we primarily focus on a different type of permutation classes, the so-called monotone grid classes. We then consider  two structural properties of a general permutation class $\cC$, called the \emph{long path 
property} (LPP) and the \emph{deep tree property} (DTP). Both these properties can be viewed as stating that 
$\cC$ contains monotone grid subclasses of a particular type. We establish lower bounds for the complexity of
\PSPPM{$\cC$} applicable to any class $\cC$ with LPP or DTP. The definitions of LPP and DTP are somewhat 
technical (see Section~\ref{sec-grid}); however, it is usually not too hard to verify whether a given class 
has these properties. Indeed, we are able to identify all the principal 
classes that have LPP, as well as all those that have DTP; see Subsection~\ref{ssec-principal}.

The LPP has already played a central part in a dichotomy result of the authors~\cite{Jelinek2020}, and implicitly also in the work of Berendsohn~\cite{BerendsohnMs} and Berendsohn et al.~\cite{Berendsohn19}.
These previous results imply that for a monotone grid class $\cC$ these properties are equivalent (assuming 
$P\neq NP$): (i) $\cC$ has LPP, (ii) $\tw_\cC(k)$ is unbounded, (iii) $\tw_\cC(k)=\Omega(\sqrt k)$, and (iv) \PPPM{$\cC$} is NP-complete.
For all we know, the equivalence might hold for an arbitrary hereditary class~$\cC$, i.e., not 
just a monotone grid class. However, we do not even know whether every class of unbounded tree-width 
has LPP. 

The DTP is a strengthening of LPP, which we introduce in this paper, with the aim of distinguishing 
the cases of \PSPPM{$\cC$} that can be solved in the subexponential time $f(k) n^{O(\sqrt k)}$ from 
those that cannot be solved in time $f(k)n^{o(k/\log k)}$. While LPP forces tree-width growth of 
order $\Omega(\sqrt k)$, DTP forces tree-width growth of order $\Omega(k/\log k)$.

Our main results show that the lower bounds on tree-width imposed by LPP and DTP are accompanied by 
the corresponding complexity lower bounds for \PSPPM{$\cC$}. With mild technical assumptions, 
we show that under ETH, the following holds for a permutation class $\cC$ (see 
Theorem~\ref{thm:main} for the precise statement):
\begin{itemize}
\item If $\cC$ has the LPP, then \PSPPM{$\cC$} cannot be solved in time 
$f(k)n^{o(\sqrt{k})}$ for any function $f$, and
\item if $\cC$ has the DTP, then \PSPPM{$\cC$} cannot be solved in time 
$f(k)n^{o(k/\log^2 k)}$ for any function~$f$.
\end{itemize}

In addition, we show that for classes with LPP, the Ahal--Rabinovich \PPM algorithm with complexity 
$n^{O(\tw(\pi))}$ is asymptotically optimal. More precisely, we show that if ETH holds, then
for a class $\cC$ with LPP, no algorithm may solve \PPPM{$\cC$} in time $f(t)n^{o(t)}$ 
for any function $f$, where $t=\tw(\pi)$ (see Theorem~\ref{thm:tw-bound}).
All these complexity lower-bounds are presented in Section~\ref{sec-hardness}.

Recall that by a result of Berendsohn et al.~\cite{Berendsohn19}, the class $\cC=\Av(321)$ has 
tree-width growth $\tw_\cC(k)=O(\sqrt{k})$, and therefore \PSPPM{$\cC$} can be solved in time $n^{O(\sqrt k)}$. 
It turns out that this class has LPP, which implies, by our results above, that $\tw_\cC(k)=\Omega(\sqrt k)$ and that \PSPPM{$\cC$} cannot be solved in time 
$f(k)n^{o(\sqrt{k})}$ for any function~$f$. In particular, both the tree-width bound and the complexity bound are tight. 

For any class $\cC$ with DTP, the tree-width lower bound $\Omega(k/\log k)$ and the complexity lower-bound 
$f(k)n^{o(k/\log^2 k)}$ both match, up to the logarithmic terms, the trivial upper bounds of $k$ and $n^{O(k)}$, respectively.

As we mentioned before, we mostly focus on monotone grid classes. We will show that for a monotone grid class $\cC$, 
both LPP and DTP can be easily characterised in terms of a certain graph associated to a monotone grid class $\cC$, called the cell graph, and that these two properties 
asymptotically determine~$\tw_\cC(\cdot)$. An earlier paper of the authors~\cite{Jelinek2020} shows 
that a monotone grid class has bounded tree-width (and hence neither LPP nor DTP) if and only if its cell graph is acyclic. We 
extend this result as follows (see Corollary~\ref{cor:grid-dichotomy}):
\begin{itemize}
\item If the cell graph of a monotone grid class $\cC$ is not acyclic but has at most one cycle in 
each component, then $\cC$ has LPP but not DTP, and $\tw_\cC(k) \in \Theta(\sqrt k)$.
\item If the cell graph of a monotone grid class $\cC$ has a component with at least two cycles, 
then $\cC$ has DTP and $\tw_\cC(k) \in \Omega(k/\log k)$.
\end{itemize}

\section{Preliminaries}\label{sec-prelims}
A \emph{permutation of length $n$} is a sequence in which each element of the set $[n] = \lbrace 1, 2, \dots, n\rbrace$ appears exactly once. When writing out short permutations explicitly, we shall 
omit all punctuation and write, e.g., $15342$ for the permutation $1,5,3,4,2$.  The \emph{permutation diagram} of $\pi$ is the set of points $S_\pi = \{(i,\pi_i);\;i\in[n]\}$ in the plane. Observe that no two points from $S_\pi$ share the same $x$- or $y$-coordinate. We say that such a set is in \emph{general position}.

For a point $p$ in the plane, we let $p.x$ denote its horizontal coordinate, and $p.y$ its vertical
coordinate. Two finite sets $S, R \subseteq \mathbb{R}^2$ in general position are 
\emph{isomorphic} if there is a bijection $f\colon S \to R$ such that for any pair of points $p \neq 
q$ of $S$ we have $f(p).x < f(q).x$ if and only if $p.x < q.x$, and $f(p).y < f(p).y$ if and only if 
$p.y < q.y$. The \emph{reduction} of a finite set $S \subseteq \mathbb{R}^2$ in general position is 
the unique permutation $\pi$ such that $S$ is isomorphic to $S_\pi$. We write $\pi = \red(S)$.

We say that a permutation $\tau$ \emph{contains} a permutation $\pi$, written $\pi\le\tau$, if 
the diagram of $\tau$ contains a subset that is isomorphic to the diagram of~$\pi$. If $\tau$ does 
not contain $\pi$, we say that it \emph{avoids}~$\pi$. A \emph{permutation class} is a set $\cC$ of 
permutations which is \emph{hereditary}, i.e., for every $\sigma\in\cC$ and every $\pi\le\sigma$, 
we have $\pi\in\cC$. For a permutation $\pi$, we let $\Av(\pi)$ denote the set of all the 
permutations that avoid $\pi$; this is clearly a permutation class. The class $\Av(21)$ of all the 
increasing permutations and the class $\Av(12)$ of all the decreasing permutations are denoted by 
the symbols $\Inc$ and $\Dec$, respectively.

We will frequently refer to symmetries that transform permutations into other permutations. For our 
purposes, it is convenient to describe these symmetries geometrically, as transformations of the 
plane acting on permutation diagrams.  We define the \emph{$m$-box} to be the set 
$(\frac{1}{2},m+\frac{1}{2}) \times (\frac{1}{2},m+\frac{1}{2})$. Observe that for every permutation 
$\pi$ of length at most $m$, the permutation diagram $S_\pi$ is a subset of the $m$-box. We view 
permutation symmetries as bijections acting on the whole $m$-box. There are eight such
symmetries, generated by:
\begin{description}
	\item[reversal] which reflects the $m$-box horizontally, i.e. the image of point $p$ is $(m + 1 - p.x, p.y)$,
	\item[complement] which reflects the $m$-box vertically, i.e. the image of point $p$ is $(p.x, m + 1 - p.y)$,
	\item[inverse] which reflects the $m$-box through its main diagonal, i.e. the image of point $p$ is $(p.y, p.x)$.
\end{description}
In particular, the reversal of a permutation $\pi=\pi_1,\dotsc,\pi_n$ is the permutation 
$\pi^r=\pi_n\pi_{n-1},\dotsc,\pi_1$, the complement of $\pi$ is the permutation 
$\pi^c=n+1-\pi_1,n+1-\pi_2,\dotsc,n+1-\pi_n$, and the inverse $\pi^{-1}$ is the permutation 
$\sigma=\sigma_1,\dotsc,\sigma_n$ such that $\sigma_i=j \iff \pi_j=i$. We also apply these 
symmetries to sets of permutations, in an obvious way: if $\Psi$ is one of the eight symmetries 
defined above and $\cC$ is a permutation class, we define $\Psi(\cC)$ as $\{\Psi(\pi);\; 
\pi\in\cC\}$.

The \emph{incidence graph} $G_\pi$ of a permutation $\pi=\pi_1,\dotsc,\pi_n$ is the graph whose vertices 
are the $n$ entries $\pi_1,\dotsc,\pi_n$, with two entries $\pi_i$ and $\pi_j$ connected by an edge 
if $|i-j|=1$ or $|\pi_i-\pi_j|=1$. In particular, the graph $G_\pi$ is a union of two paths, one of 
them visiting the entries of $\pi$ in left-to-right order, and the other in top-to-bottom order. We 
let $\tw(\pi)$ denote the tree-width of~$G_\pi$. 

\paragraph*{Monotone grid classes} 
An important type of permutation classes are the so-called monotone grid-classes, which we now 
define. A \emph{gridding matrix of size $k\times\ell$} is a matrix $\cM$ with $k$ columns and $\ell$ 
rows, whose every entry is a permutation class. A \emph{monotone gridding matrix} is a gridding 
matrix whose every entry is one of the three classes $\emptyset$, $\Inc$ or $\Dec$. Note that to be 
consistent with the Cartesian coordinates that we use to describe permutation diagrams, we will 
number the rows of a matrix from bottom to top, and we give the column coordinate as the first one. 
In particular, $\cM_{i,j}$ denotes the entry in column $i$ and row $j$ of the matrix $\cM$, with 
$1\le i\le k$ and $1\le j\le \ell$.

Let $\pi$ be a permutation of length~$n$. A \emph{$(k\times\ell)$-gridding} of $\pi$ is a pair of 
weakly increasing sequences $1=c_1\le c_2\le\dotsb\le c_{k+1}=n+1$ and $1=r_1\le r_2\le\dotsb\le 
r_{\ell+1}=n+1$. For $i\in[k]$ and $j\in[\ell]$, the \emph{$(i,j)$-cell} of the gridding of $\pi$ is 
the set of points $p\in S_\pi$ satisfying $c_i\le p.x<c_{i+1}$ and $r_j\le p.y< r_{j+1}$. Note that 
each point of the diagram $S_\pi$ belongs to a unique cell of the gridding. A permutation $\pi$ 
together with a gridding $(c,r)$ forms a \emph{gridded permutation}.

Let $\cM$ be a gridding matrix of size $k\times\ell$. We say that the gridding of $\pi$ is 
an \emph{$\cM$-gridding} if for every $i\in[k]$ and $j\in [\ell]$, the subpermutation of $\pi$ 
induced by the points in the $(i,j)$-cell of the gridding of $\pi$ belongs to the class $\cM_{i,j}$.

We let $\Grid(\cM)$ denote the set of permutations that admit an $\cM$-gridding. This is clearly a 
permutation class. A \emph{monotone grid class} is any permutation class $\Grid(\cM)$ for a monotone 
gridding matrix~$\cM$. 

The \emph{cell graph} of a gridding matrix $\cM$, denoted $G_\cM$, is the graph whose vertices are 
all the pairs $(i,j)$ for which $\cM_{i,j}$ is an infinite permutation class. Two vertices are 
adjacent if they appear in the same row or the same column of $\cM$, and there is no other cell 
containing an infinite class between them.
 See Figure~\ref{fig:grid-class}.
A \emph{proper-turning path} in $G_\cM$ is a path $P$ such that no three vertices of $P$ share the same row or column.

\begin{figure}
  \centering
  \raisebox{-0.5\height}{\includegraphics[width=0.45\textwidth]{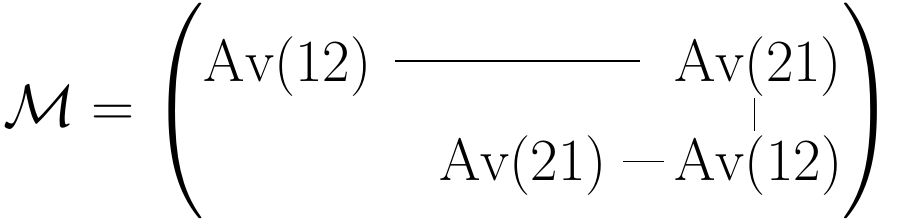}}
  \hspace{0.5in}
  \raisebox{-0.5\height}{\includegraphics[scale=0.5]{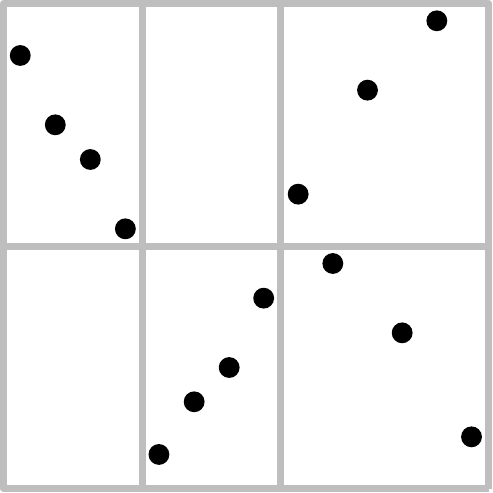}}
  \caption{A monotone gridding matrix $\cM$ on the left and a permutation equipped with an $\cM$-gridding on the right. Empty 
    entries of $\cM$ are omitted and the edges of $G_\cM$ are drawn in $\cM$.}
  \label{fig:grid-class}
\end{figure}

\paragraph*{Grid transforms and orientations} Let $\pi$ be a permutation of length $n$ with a 
$(k\times\ell)$-gridding $(c,r)$, where $c=(c_1,\dotsc,c_{k+1})$ and $r=(r_1,\dotsc,r_{\ell+1})$. 
The \emph{reversal of the $i$-th column} of $\pi$ is the operation that transforms $\pi$ into a new 
permutation $\pi'$ by taking the rectangle $[c_i,c_{i+1}-1]\times[1,n]$ and flipping it along its
vertical axis, thus producing the diagram of a new permutation $\pi'$. Equivalently, $\pi'$ is 
created from $\pi$ by reversing the order of the entries of $\pi$ at positions $c_i, c_i+1,\dotsc,
c_{i+1}-1$.  We view $\pi'$ as a gridded permutation, with the same gridding $(c,r)$ as~$\pi$.

Similarly, the \emph{complementation of the $j$-th row} transforms the diagram of $\pi$ 
by flipping the rectangle $[1,n]\times[r_j,r_{j+1}-1]$ along its horizontal axis, 
producing the diagram of a new gridded permutation~$\pi'$.

We may similarly apply reversals to the columns of a gridding matrix $\cM$ and complements to its 
rows. Reversing the $i$-th column of $\cM$ produces a new gridding matrix, in which all the classes 
in the $i$-th column of $\cM$ are replaced by their reversals. Row complementation of a gridding 
matrix is defined analogously. Note that a column reversal or a row complementation in a gridded 
permutation or in a gridding matrix is an involution, i.e., repeating the same operation twice 
restores the original permutation or matrix. Note also that when we perform a sequence of column 
reversals and row complementations, then the end result does not depend on the order in which the 
operations were performed.

To describe succinctly a sequence of row and column operations, we introduce the notion of 
\emph{$(k\times\ell)$-orientation}, which is a pair of functions $\cF=(f_c,f_r)$ with 
$f_c\colon[k]\to\{-1,1\}$ and $f_r\colon[\ell]\to\{-1,1\}$. \emph{Applying} the orientation $\cF$ to a 
$(k\times\ell)$-gridded permutation $\pi$ produces a new gridded permutation $\cF(\pi)$ with the 
same gridding as~$\pi$, obtained by reversing each column $i$ such that $f_c(i)=-1$ and 
complementing each row $j$ such that $f_r(j)=-1$. The application of $\cF$ to a gridding matrix 
$\cM$ is defined analogously, and produces a gridding matrix denoted $\cF(\cM)$. Note that
$(c,r)$ is an $\cM$-gridding of $\pi$ if and only if it is an $\cF(\cM)$-gridding of $\cF(\pi)$.

An orientation $\cF$ is a \emph{consistent orientation} of a monotone gridding matrix $\cM$, if 
every nonempty entry of $\cF(\cM)$ is equal to~$\Inc$. As an example,  the matrix 
$\left(\begin{smallmatrix} \Dec & \Inc\\ \Inc &\Dec \end{smallmatrix}\right)$ has a consistent 
orientation acting by reversing the first column and complementing the first row. On the other hand, 
the matrix $\left(\begin{smallmatrix} \Dec & \Inc\\ \Inc &\Inc \end{smallmatrix}\right)$ has no 
consistent orientation, since applying any orientation to this matrix yields a matrix with an odd 
number of $\Dec$-entries. 

The following lemma, due to Vatter and Waton~\cite{Vatter2011}, will later be useful.
\begin{lemma}\label{lem-acyclictrans}
Every monotone gridding matrix whose cell graph is acyclic has a consistent orientation.
\end{lemma}

\paragraph*{Tile assembly} In the hardness reductions that we are about to present, we frequently 
need to construct permutations whose diagrams have a natural $k\times\ell$ grid-like structure. We 
describe such a diagram by taking each cell individually and describing the points inside it. For 
such a description, it is often convenient to assume that each cell has its own coordinate system 
whose origin is near the bottom-left corner of the cell. This allows us to describe the coordinates 
of the points inside the cell without referring to the position of the cell within the whole 
permutation diagram. In effect, we describe the diagram of the gridded permutation by first 
constructing a set of independent `tiles' $\cT_{i,j}$ for $i\in[k]$ and $j\in[\ell]$ of the same 
size, and then translating each tile $\cT_{i,j}$ to column $i$ and row $j$ of the diagram.
On top of that, we often need to apply an orientation to the gridded permutation whose diagram we 
constructed.

We now describe the whole procedure more formally. Fix an integer $m$ and recall that an $m$-box is 
a square of the form $(\frac12,m+\frac12)\times (\frac12, m+\frac12)$. An $m$-tile is a finite set 
of points inside the $m$-box. Note that the coordinates of the points in the tile may not be 
integers. A $(k\times \ell)$-family of $m$-tiles is a collection $(\cT_{i,j};\; i\in[k], 
j\in[\ell])$ where each $\cT_{i,j}$ is an $M$-tile. Let $\cF$ be a $(k\times\ell)$-orientation. The 
\emph{$\cF$-assembly} of the family $(\cT_{i,j};\; 
i\in[k], j\in[\ell])$ is the gridded permutation obtained as follows. 

First, we translate each tile $\cT_{i,j}$ by adding $m(i-1)$ to each horizontal coordinate and 
$m(j-1)$ to each vertical coordinate. Thus, the $m$-tiles will be disjoint. If the union of the 
translated tiles is not in general position, we rotate it slightly clockwise to reach general 
position. Notice that we can do so without changing the relative position of any pair of points that were already in general position. This yields a point set isomorphic to a unique permutation~$\pi$.
See Figure~\ref{fig:F-assembly}.
Additionally, $\pi$ has 
a natural gridding whose cells correspond to the translated tiles. To finish the construction, we 
apply the orientation $\cF$ to $\pi$, obtaining the gridded permutation $\cF(\pi)$, which is the 
$\cF$-assembly of the family of tiles $(\cT_{i,j};\; i\in[k], j\in[\ell])$.

\begin{observation}
Let $(\cT_{i,j};\; i\in[k], j\in[\ell])$ be a family of tiles, let $\cF$ be an orientation, and let 
$\cM$ be a gridding matrix such that $\cT_{i,j}$ is isomorphic to a permutation from the 
class $\cM_{i,j}$. Then the $\cF$-assembly of the family of tiles  $(\cT_{i,j};\; i\in[k], 
j\in[\ell])$ is a permutation from the class $\Grid(\cF(\cM))$. 
\end{observation}

\begin{figure}
  \centering
  $\begin{array}{cc}
    T_{1,2} = \vcenter{\hbox{\includegraphics[scale=0.45]{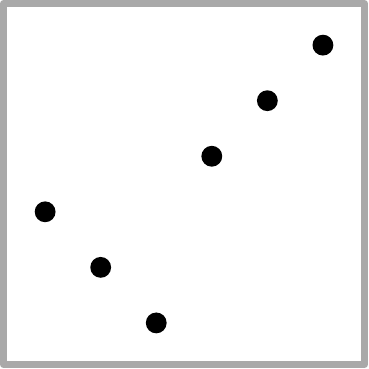}}}  & 
    T_{2,2} = \vcenter{\hbox{\includegraphics[scale=0.45]{T-tile}}}\\[0.3in]
    T_{1,1} = \vcenter{\hbox{\includegraphics[scale=0.45]{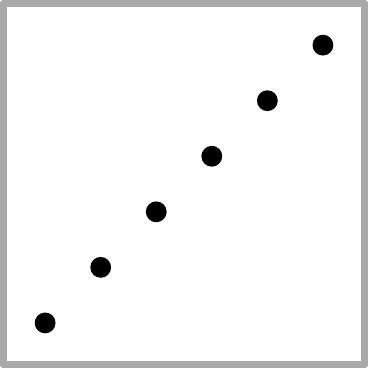}}} &
    T_{2,1} = \vcenter{\hbox{\includegraphics[scale=0.45]{T-tile}}} 
  \end{array}$
  \hspace{0.2in}
  $\longrightarrow$
  \hspace{0.2in}
  \raisebox{-0.5\height}{\includegraphics[scale=0.45]{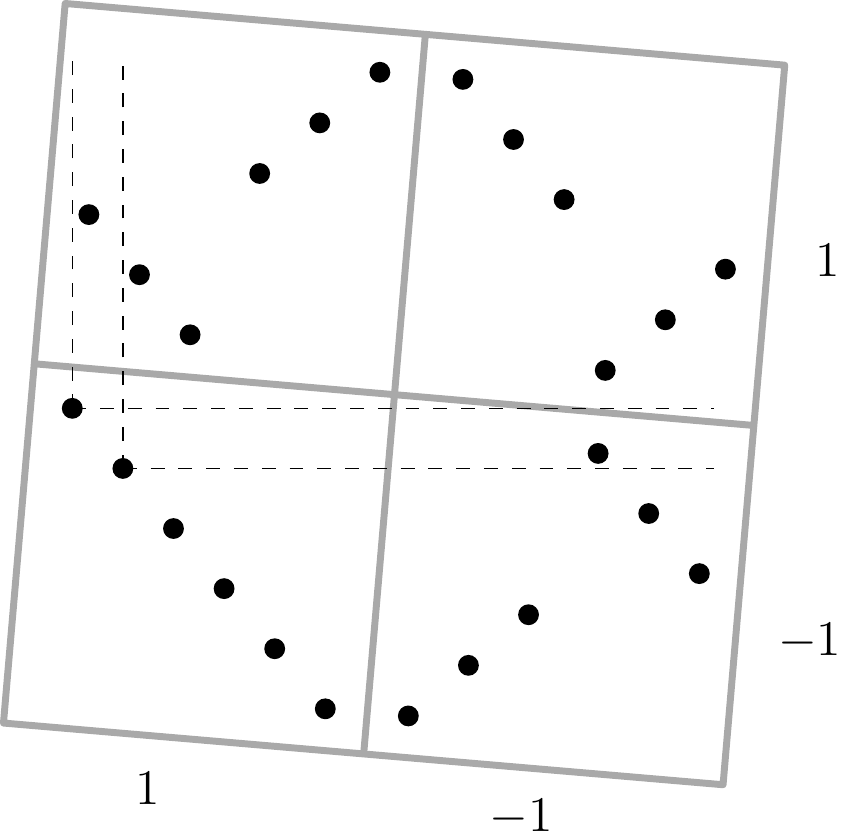}}
  \caption[Example of an $\cF$-assembly]{
    A $2 \times 2$ family of tiles $\cT$ on the left and its $\cF$-assembly on the right for 
    a $2 \times 2$ orientation $\cF$ given next to each row and column on the right. General position is attained by rotating the resulting point set clockwise. The dashed lines indicate relative positions of two particular points.
  }\label{fig:F-assembly}
\end{figure}

\section{Tree-width bounds}\label{sec-grid}

\subsection{Width of monotone grid classes}
\label{subsec:width}

We say that a permutation class $\cC$ has the \emph{long path property} (LPP) if for every $k$ the class $\cC$ contains a monotone grid subclass whose cell graph is a path of length~$k$. The next proposition builds upon the ideas of Berendsohn et al.~\cite{Berendsohn19}, who proved a similar result for the class $\Av(321)$ using the fact that this class contains a staircase-shaped grid path of arbitrary length.

\begin{proposition}\label{pro-path}
If a permutation class $\cC$ has the LPP then $\tw_\cC(n) \in \Omega(\sqrt{n})$.
\end{proposition}
\begin{proof}
First, we show that $\cC$ contains for every $k$ a grid subclass whose cell graph is a proper-turning path of length $k$, i.e. a path in which no three consecutive vertices are in the same row or column of the gridding. For the 
contrary, assume that there is $\ell$ such that $\cC$  does not contain such path of 
length $\ell$. The LPP then implies that $\cC$ contains for every $t$ a class $\Grid(\cM$) where 
$\cM$ is either a  $1\times t$ or $t \times 1$ matrix without empty entries. However, any such matrix 
of dimensions $1 \times n$ or  $n \times 1$ contains all permutations of length $n$ and thus, $\cC$ 
must 
actually be the class of all permutations that contains all possible proper turning paths.

So we can suppose that there is a monotone gridding matrix $\cM$ such that $\cM$ is a proper-turning path $v_1, \ldots, v_{2m-1}$ of length $2m-1$  and $\Grid(\cM)$ is contained in $\cC$. We explicitly construct a permutation $\pi \in \Grid(\cM)$ such that $G_\pi$ contains an $m \times m$ grid graph as a subgraph. The claim then follows since the tree-width of $m\times m$ grid graph is exactly $m$.

For $i \in [m]$ and $j \in [i]$, let
\begin{align*}
p_{i,j} = (m + 2j - i - 1, m + 2j -i - 1),\quad p_{2m-i, j} = p_{i,j}.
\end{align*}

We define a family of $2m$-tiles $\cP$ by setting 
$\cP_{v_i}$ to be the set of points $p_{i,j}$ for all possible choices of $j$.

Let $\cF$ be a consistent orientation of $\cM$ guaranteed by 
Lemma~\ref{lem-acyclictrans} and let $\pi$ be the $\cF$-assembly of $\cP$. The sets $P_{v_i}$ were defined in such a way that for 
every $i$ the points in $P_{v_{2i}}$ have both coordinates odd whereas the points in $P_{v_{2i+1}}$ have 
both coordinates even. Since $\cM$ is a proper 
turning path, there are always at most two non-empty tiles sharing the same row or column in $\pi$ 
and in such case they correspond to neighboring vertices of the path.
Moreover, if they share a common row, then the $y$-coordinates of their points are interleaved, and if they share a common column, the same holds for the $x$-coordinates.

\begin{figure}
  \centering
  \raisebox{-0.5\height}{\includegraphics[scale=0.6]{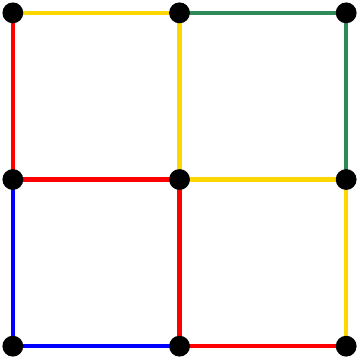}}
  \hspace{1in}
  \raisebox{-0.5\height}{\includegraphics[scale=0.6]{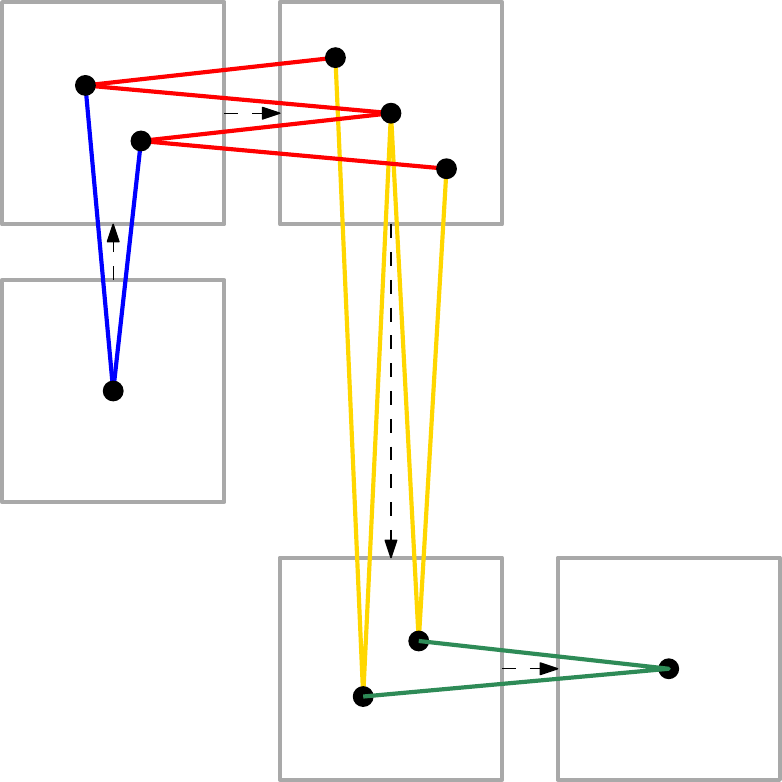}}
  \caption{Illustration of the proof of Proposition~\ref{pro-path}. Embedding a $3 \times 3$ grid graph (left) into a permutation from a monotone grid class whose cell graph is a path of length 5 (right).
  }
  \label{fig:lpp-grid}
\end{figure}

It remains to show that $G_\pi$ contains an $m \times m$ grid graph as a subgraph. Let $s_{i,j}$ be the image of $p_{i,j}$ under the $\cF$-assembly. We claim that we can map consecutive diagonals of the grid to the tiles $P_{v_i}$. See Figure~\ref{fig:lpp-grid}.  More precisely, for $x, y \in [m]$ set
\[g_{x,y} = \begin{cases}
s_{x+y-1, x} &\text{if $x + y \leq m + 1$,} \\
s_{x+y-1, m-y+1} &\text{otherwise.}
\end{cases}\]

We start by showing that for any $i \in [m-1]$, there is an edge between $s_{i,j}$ and $s_{i+1,j}$, and also between $s_{i,j}$ and $s_{i+1,j+1}$. This follows since the points of $P_{v_i}$ and $P_{v_{i+1}}$ have their $x$- or $y$-coordinates interleaved and there is no other tile occupying their shared row or column. Due to symmetry, it holds that for $i > m$, there is an edge between  $s_{i,j}$ and $s_{i-1,j}$ and also between $s_{i,j}$ and $s_{i-1,j+1}$.

If we take $x,y \in [m]$ such that $x + y \le m$ (i.e. $g_{x,y}$ lies below the anti-diagonal of the grid), the fact proved in the previous paragraph directly translates to the existence of edges between $g_{x,y}$ and $g_{x+1,y}$ and between $g_{x,y}$ and $g_{x,y+1}$. On the other hand for $x,y \in [m]$ such that $x + y \ge m + 2$, the points $g_{x,y}$, $g_{x-1,y}$ and $g_{x,y-1}$ translate to $s_{x+y-1,x}$, $s_{x+y-2,x-1}$ and $s_{x+y-2,x}$. Therefore, in this case there are edges between $g_{x,y}$ and $g_{x-1,y}$ and between $g_{x,y}$ and $g_{x,y-1}$. This concludes the proof as any edge in the $m \times m$ grid graph is of one of the two types whose existence we proved.
\end{proof}

It turns out that there is a large family of monotone grid classes for which $\tw_\cC \in \Theta(\sqrt{n})$, namely every monotone grid class whose cell graph is not acyclic yet it does not contain two connected cycles.
%We include the complete proof in the full version, and here we only briefly describe its main ideas.

\begin{theorem}\label{thm-unicyclic}
If $\cM$ is a monotone gridding matrix whose cell graph is connected and contains a single cycle then $\tw_{\Grid(\cM)}(n) \in \Theta(\sqrt{n})$.
\end{theorem}

Before we can prove an upper bound on the tree-width of $\Grid(\cM)$ for classes with a single cycle, we need to take a quick detour and introduce a different type of permutation classes defined using gridding matrices. Let $\cM$ be a monotone gridding matrix. The standard figure of $\cM$, denoted by $\Lambda_\cM$,  is the point set in the plane consisting of:
\begin{itemize}
  \item the open segment from $(i-1, j-1)$ to $(i,j)$ for every $i,j$ such that $\cM_{i,j} = \, \Inc$, and
  \item the open segment from $(i-1, j)$ to $(i,j-1)$ for every $i,j$ such that $\cM_{i,j} = \,\Dec$.
\end{itemize}
The \emph{geometric grid class} of $\cM$, denoted by $\Geom(\cM)$, is the set of all permutations isomorphic to finite subsets of $\Lambda_\cM$ in general position. Clearly, $\Geom(\cM) \subseteq \Grid(\cM)$ for any monotone gridding matrix $\cM$. However, equality holds for acyclic gridding matrices.

\begin{proposition}[Albert et al.\cite{Albert2013}]
  \label{pro-geom-grid}
  If $\cM$ is an acyclic monotone gridding matrix then $\Geom(M) = \Grid(\cM)$.
  In fact, given any $\cM$-gridding of a permutation $\pi$ we can find a subset of $\Lambda_\cM$ that is isomorphic to $S_\pi$ and respects the given $\cM$-gridding.
\end{proposition}

Let $G = (V, E)$ be a graph together with a linear ordering $\le_\cL$ of its vertices. We say that a set of edges $F \subseteq E$ is \emph{non-overlapping with respect to $\le_\cL$} if for every two edges $(u_1, u_2), (v_1, v_2) \in F$ it holds that either $u_\alpha \le_\cL v_\beta$ for all possible choices of $\alpha, \beta \in \{1,2\}$ or $v_\beta \le_\cL u_\alpha $ for all possible choices of $\alpha, \beta \in \{1,2\}$. In other words, we forbid two edges whose endpoints are ordered as $u_1 \le_\cL v_1 \le_\cL u_2 \le_\cL v_2$ or $u_1 \le_\cL v_1 \le_\cL v_2 \le_\cL u_2$.
The following lemma proves that for any permutation $\pi$ from acyclic monotone grid class, there is a linear ordering of its vertices such that the edges of $G_\pi$  can be partitioned into few non-overlapping sets.
This ordering is later used for drawing $G_\pi$ on a surface with a small genus.

Let $\cM$ be a $k\times\ell$ gridding matrix and let $\pi$ be an $\cM$-gridded permutation. An edge $\{\pi_i,\pi_j\}$ of the incidence graph $G_\pi$ is \emph{horizontal} if $|i-j|=1$, and it is \emph{vertical} if $|\pi_i-\pi_j|=1$. Thus, the horizontal edges form a left-to-right path, and the vertical ones form a bottom-to-top path. A horizontal edge is said to be \emph{exceptional} (with respect to the given gridding) if its vertices belong to different columns of the gridding, and a vertical edge is \emph{exceptional} if its vertices belong to different rows. There are therefore $k-1$ exceptional horizontal edges and $\ell-1$ exceptional vertical edges, hence at most $k+\ell-2$ exceptional edges overall.

\begin{lemma}
  \label{lem:acyclic-order}
  Let $\cM$ be an acyclic monotone $k \times \ell$ gridding matrix equipped with a consistent $k \times \ell$ orientation $\cF$, and let $\pi$ be a permutation from $\Grid(\cM)$ equipped with an $\cM$-gridding. There exists a linear order $\le_\cL$ on the points of $\pi$ such that
  \begin{enumerate}[(a)]%[label=(\alph*)]
    \item the points of the $i$-th column of the $\cM$-gridding induce a suborder given by the increasing order of their $x$-coordinates if $f_c(i) = 1$ and the decreasing order otherwise, \label{cond-order-col}
    \item the points of the $j$-th row of the $\cM$-gridding induce a suborder given by the increasing order of their $y$-coordinates if $f_r(j) = 1$ and the decreasing order otherwise, and \label{cond-order-row}
    \item the edges of the graph $G_\pi$ can be partitioned into at most $2k + 2\ell -2$ non-overlapping sets with respect to $\le_\cL$. \label{cond-edge-partition}
  \end{enumerate}
  
\end{lemma}

\begin{proof}
%  We say that an edge $e$ in $G_\pi$ is \emph{exceptional} if it connects two points that occupy
% different rows and different columns of the $\cM$-gridding. It is easy to observe that there are at
% most $k + \ell - 2$ exceptional edges. 
  We partition the exceptional edges into at most $k + \ell - 2$ singleton sets which are trivially non-overlapping with respect to arbitrary linear order. 
  
  Let $S$ be the subset of the standard figure $\Lambda_\cM$ isomorphic to $S_\pi$ that exists by Proposition~\ref{pro-geom-grid}. Let $g\colon S_\pi \to S$ be the witnessing bijection. For a point $p \in S_\pi$ such that $p$ belongs to the $(i,j)$-cell of the $\cM$-gridding, we define the rank of $p$ as the distance of the point $g(p)$ from the point
  \[\left(i - \frac{1}{2} - \frac{f_c(i)}{2},\; j - \frac{1}{2} - \frac{f_r(j)}{2}\right).\]
  Less formally, we define the rank to be the distance to a specific corner of the rectangle $[i-1, i] \times [j-1,j]$, depending on the orientation of the $(i,j)$-cell.
  Observe that we can always choose $S$ in a way such that the ranks are pairwise different. That allows us to define $\le_\cL$ as the linear order given by the ranks.
  
  The conditions (\ref{cond-order-col}) and (\ref{cond-order-row}) follow straightforwardly from the consistency of the orientation $\cF$ and our choices of suitable corners used for computing the ranks. Moreover, all the non-exceptional horizontal edges of $G_\pi$ can be partitioned into $k$ sets, each connecting points in a single column. And since we already know that points of a single column are ordered accordingly in $\cL$, each such set is non-overlapping. Using the same argument, we split the vertical non-exceptional edges of $G_\pi$ into $\ell$ non-overlapping sets, one for each row of the $\cM$-gridding. Together with the exceptional edges, we obtain the partition of all edges into at most $2k + 2\ell -2$ non-overlapping sets as promised.
\end{proof}

Finally, the last piece missing before the actual proof of Theorem~\ref{thm-unicyclic} is the following lemma about graphs drawn on surfaces with few crossings, which is proved using standard methods~\cite{DEW}.

\begin{lemma}
  \label{lem:drawing}
  If $G$ is a graph on $n$ vertices that can be drawn on a surface with Euler genus $g$ with $O(n)$ crossings, then $\tw(G) \in O(\sqrt{g n})$.
\end{lemma}
\begin{proof}
  Let $G = (V,E)$ be a graph on $n$ vertices together with a drawing on a surface with Euler genus $g$.
  We assume that no three edges cross in a single point.
  We define a graph  $G'$ in the following way.
  We replace each crossing with a new vertex of degree 4, and we split each $e \in E$ into its consecutive segments between endpoints and crossings.
  It follows from the assumptions that $G'$ has $O(n)$ vertices and moreover, it can be drawn on the surface with Euler genus $g$ without any crossings.
  It follows from the work of Gilbert, Hutchinson and Tarjan~\cite{Gilbert1984}, that $\tw(G') \in O(\sqrt{gn})$.
  
  Using a tree decomposition $(T, \beta)$ of $G'$ with width $t$, we define a tree decomposition $(T, \beta')$ of $G$ with width $4t+3$ by replacing every vertex corresponding to a crossing with the four endpoints of the two edges participating in this particular crossing.
  Clearly, every vertex belongs to some bag $\beta(v)$, the same holds for every edge, and the size of each bag is at most $4t+4$.
  It is also not hard to check that each vertex induces a connected subtree in $T$.
\end{proof}

\begin{proof}[Proof of Theorem~\ref{thm-unicyclic}]
  First, we show that $\tw_{\Grid(\cM)}(n) \in \Omega(\sqrt{n})$. It has been previously proved by the authors in~\cite[Lemma 3.5]{Jelinek2020} that a cycle in a grid class implies the LPP. The lower bound readily follows from Proposition~\ref{pro-path}.
  
  For the upper bound, let $\pi$ be a permutation of $\Grid(\cM)$ with a given $\cM$-gridding. 
  We start by removing all exceptional edges.
  As we observed, there are at most $k + \ell - 2$ exceptional edges.
  Let $G' = (V', E')$ be the graph obtained from $G_\pi$ by removing them from $G_\pi$.
  It is sufficient to show that $\tw(G') \in O(\sqrt{n})$, as adding back the exceptional edges increases the tree-width at most by a constant.
  
  We aim to show that $G'$ can be drawn on a surface of Euler genus 1 with $O(n)$ total crossings. Suppose that $c_1, c_2, \dots, c_m$ are the entries of $\cM$ that lie on its only cycle in this order. Let $a_i, b_i$ be the coordinates of the entry $c_i$, i.e., $c_i$ lies in the $a_i$-th column and $b_i$-th row.
  The cell graph $G_\cM$ consists of the cycle and trees that are attached to it. If we remove all the edges that participate in the cycle, we end up with $m$ trees $T_1, \dots, T_m$ called \emph{tendrils} such that the tree $T_i$ contains the entry $c_i$.
  
  We now define two functions $f^\star_c, f^\star_r\colon [m] \to \lbrace -1, 1\rbrace$ that will capture an almost consistent orientation of the entries $c_1, \dots, c_m$. Let $\cM'$ be the gridding matrix obtained from $\cM$ by removing everything except for the entries on the cycle and then additionally removing $c_1$. It follows from our assumption about $c_1$ that $\cM'$ is acyclic and thus, it has a consistent orientation $\cF' = (f'_c, f'_r)$. For every $i \ge 2$, we set $f^\star_c(i), f^\star_r(i)$ to be the values $f'_c(a_i), f'_r(b_i)$. Additionally, we define $f^\star_c(1) = f'_c(a_1)$ and then we set $f^\star_r(1)$ such that $f^\star_c(1) \cdot f^\star_r(1) = 1$ if and only if $c_1$ is an increasing entry. In this way, it holds for every $i \in [m]$ that $f^\star_c(i) \cdot f^\star_r(i) = 1$ if and only if $c_i$ is increasing. Moreover, if $c_i$ and $c_j$  share a common column then $f^\star_c(i) = f^\star_c(j)$. And finally, if $c_i$ and $c_j$  share a common row then also $f^\star_r(i) = f^\star_r(j)$ as long as $i$ and $j$ are different from 1.
  
 Now for $i \in [m]$, let $\cM_i$ be the gridding matrix obtained from $\cM$ by removing everything 
except the tendril $T_i$ and let $\pi^i$ be the subpermutation of $\pi$ induced by the non-empty 
entries of $\cM_i$. Since $\cM_i$ is acyclic, it has a consistent orientation $\cF^i = (f^i_c, 
f^i_r)$. Moreover, we can assume that $f^i_c(a_i) = f^\star_c(i)$ and $f^i_r(b_i) = f^\star_r(i)$ as 
otherwise we could just flip all the signs in~$\cF^i$. Applying Lemma~\ref{lem:acyclic-order} on 
$\cM_i$, $\cF^i$ and $\pi^i$, we obtain a linear order $\le_{\cL_i}$ on the points of~$\pi^i$.
  
  We are ready to describe the drawing of $G'$. In order to simplify our arguments, we draw $G'$ on an infinite cylinder $[0, m-1] \times \mathbb{R}$ where we implicitly work with the $x$-coordinates in arithmetic modulo $m$.
  Such a drawing can then easily be transformed into a drawing on the plane with the same number of crossings via a suitable projection.
  We call the vertical line $x = i$ the \emph{$i$-th meridian}.
  First, let us describe the position of the vertices. For every $i \in [m]$, the vertices corresponding to points of $\pi^i$ are placed on $i$-th meridian in the order $\le_{\cL_i}$ from bottom to top. The actual distances in between them do not matter.
  
  We split the edges of $E'$ into two groups. We call any edge that connects two points of $\pi^i$ an \emph{inner} edge, and every edge that connects a point in $\pi^i$ with a point in $\pi^j$ for $i \neq j$ an \emph{outer} edge. We start by drawing the inner edges. We first draw every inner edge as a circular arc connecting its two endpoints and, afterwards, we horizontally shrink all these edges so that they occupy only a narrow band around the $i$-th meridian.
  
Let us count the number of crossings between two inner edges. Due to the shrinking step, two inner 
edges can intersect only if they both connect two points of $\pi^i$ for some~$i$. We fix $i \in 
[m]$. Notice that the edges of $E'$ induced by $\pi^i$ form a subset of the edges of $G_{\pi^i}$. By 
part~(\ref{cond-edge-partition}) of Lemma~\ref{lem:acyclic-order}, these edges can be partitioned into 
at most $2k + 2\ell - 2$ sets $E^i_1, E^i_2, \dots, E^i_{2k + 2\ell - 2}$ that are all 
non-overlapping with respect to $\le_{\cL_i}$. The non-overlapping property implies that any edge $e 
\in E^i_j$ cannot intersect any other edge from $E^i_j$, and it can intersect at most two 
edges from every $E^i_{j'}$ for 
$j' \neq j$. Therefore, each edge participates in at most $4k + 4\ell - 6 = O(1)$ crossings. Summing 
over all choices of $i$, there are at most $O(n)$ crossings between two inner edges.
  
  Now, we describe the drawing of the outer edges.
  Observe that the outer edges can be partitioned into at most $m$ sets depending on the column or row that is shared by both their endpoints.
If an outer edge connects points of $\pi^i$ and $\pi^j$ such that $|i-j| = 1$, then we draw it as 
the  straight-line segment between the two vertices that does not cross any meridian. We need to be 
more careful with the remaining edges. We draw an edge that connects points of $\pi^i$ and $\pi^j$ 
for $|i-j| > 1$ as a polyline consisting of straight-line segments between the $h$-th and $(h+1)$-th 
meridian for every $h \in \{i, \dots, j-1\}$.
  Suitably choosing the points on each meridian, we can guarantee that the segments between $h$-th and $(h+1)$-th meridian do not intersect as long as $\le_{\cL_h}$ and $\le_{\cL_{h+1}}$ order the points in $c_h$ and $c_{h+1}$ consistently.
  
  If the whole matrix $\cM$ possesses a consistent orientation, then Lemma~\ref{lem:acyclic-order} parts (\ref{cond-order-col}) and (\ref{cond-order-row}) imply that no two outer edges intersect as the orders of cells in a given row or a column all agree. Otherwise, this is true for every column and every row except for the row $b_1$. Assuming without loss of generality that $c_1$ shares a common row with $c_2$, the order of points of $\pi^1$ and $\pi^2$ on the 1st and 2nd meridian are reversed.
  In such case, the opposite happens and every two segments of outer edges in the band between 1st and 2nd meridian intersect.
  This is, however, easily fixed by adding a cross-cap in between these two meridians such that every segment of an outer edge lying in this band crosses through it.
  Thus, we obtained a drawing of $G'$ either on a plane or on a projective plane.
  See Figure~\ref{fig-unicyclic}.
  
  \begin{figure}
    \centering
    \includegraphics[width=0.85\textwidth]{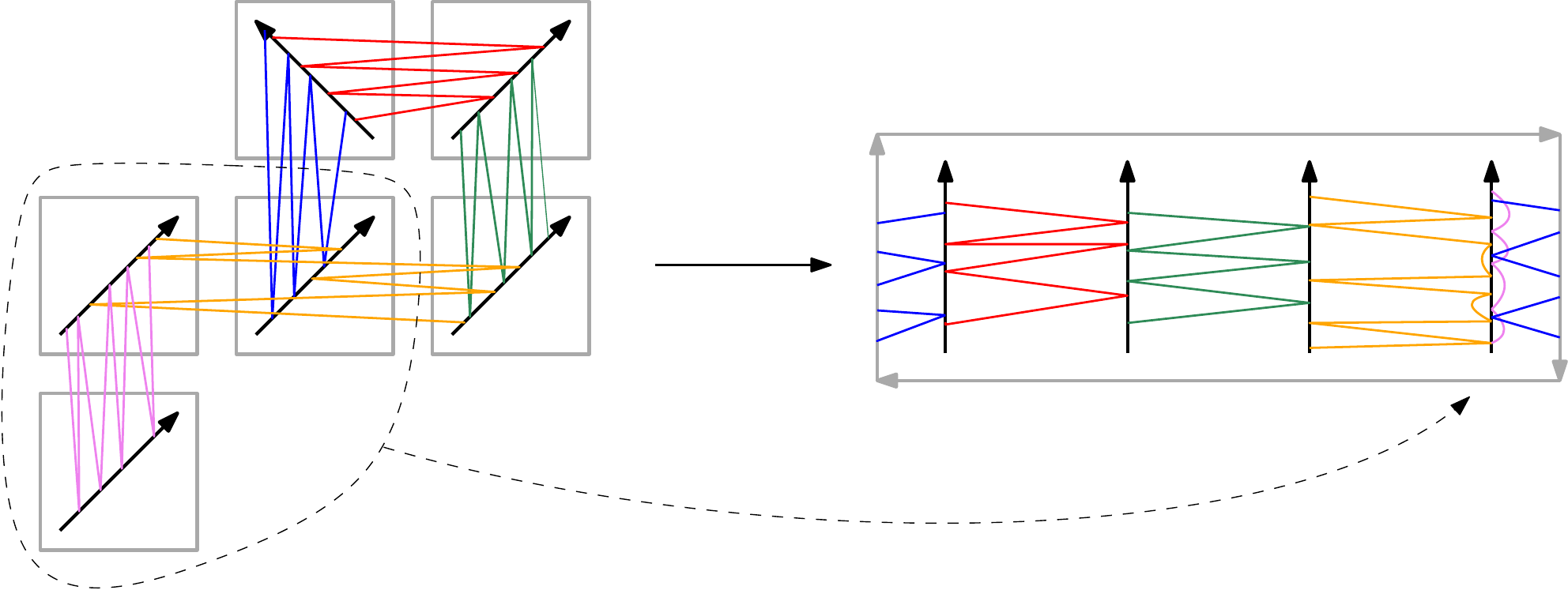}
    \caption{A schematic drawing of $G_\pi$ for $\pi$ from a unicyclic grid class on the projective plane.
      Instead of drawing the specific points of $\pi$, we place arrows to indicate the orientation of each cell.
      Different color is used for each set of edges that share a single row or column, and the exceptional edges are omitted.
    }\label{fig-unicyclic}
  \end{figure}
  
  It remains to count the number of intersections formed between pairs consisting of an inner and an outer edge.
  Fix an outer edge $e$ and recall that $e$ is a polyline formed by at most $m$ segments.
  A segment between the $i$-th and $(i+1)$-th meridian can intersect only with inner edges of $\pi^i$ in the near vicinity of the $i$-th meridian and with inner edges of $\pi^{i+1}$ in the near vicinity of the $(i+1)$-th meridian.
  Moreover, it can intersect at most one edge from any non-overlapping set with respect to $\le_{\cL_i}$ or $\le_{\cL_{i+1}}$.
  Thus, it follows from Lemma~\ref{lem:acyclic-order}~\ref{cond-edge-partition} that each segment intersects at most $4k+4\ell-4$ edges and $e$ intersects at most $m \cdot (4k+4\ell-4) = O(1)$ inner edges in total.
  Putting it all together, there are $O(n)$ crossings in total and by Lemma~\ref{lem:drawing} we get that $\tw(G') \in O\bigl(\sqrt{n}\bigr)$.
\end{proof}

For integer constants $c$ and $d$, a \emph{$c$-subdivided binary tree of depth $d$} is a graph 
obtained from a binary tree of depth $d$ by replacing every edge by a path of length at most~$c$.
We say that a permutation class $\cC$ has the \emph{deep tree property} (DTP) if there is a constant $c$ such that for every $d$, the class $\cC$ contains a monotone grid subclass whose cell graph is a 
$c$-subdivided binary tree of depth~$d$.
Observe that DTP straightforwardly implies LPP.
We say that a class $\cC$ has \emph{near-linear width} if $\tw_\cC(n) \in \Omega(n/\log n)$.

\begin{proposition}\label{pro-tree}
If a permutation class $\cC$ has the DTP, then it has near-linear width.
\end{proposition}
\begin{proof}
Inspired by the approach of Berendsohn~\cite{BerendsohnMs}, we want to show that for a graph $G$ 
of large tree-width, we can find a permutation $\sigma \in \cC$ such that $G_\sigma$ contains $G$ 
as a minor while the length of $\sigma$ exceeds the size of $G$ by at most a logarithmic factor.

To that end, fix an arbitrary graph $G$ with vertex set $V_G=[n]$ and edges $\{e_1, \ldots, e_m\}$ where $e_i = \{a_i, b_i\}$. Let $\cM$ be a monotone gridding matrix such that $\Grid(\cM) \subseteq \cC$ and the cell graph of $\cM$ is a $c$-subdivided binary tree with exactly $m$ leaves. Let $r$ denote the root of this tree. It follows that the tree has maximal depth at most $c(\log m + 1)$. We turn $G_\cM$ into an oriented graph by orienting all edges consistently away from $r$. For any vertex $v$ of the tree, the \emph{descendants of $v$}, denoted by $D(v)$, are all the out-neighbors of $v$.

We assign a set $A_w \subseteq V_G$ to each vertex $w$ of the tree. First, we arbitrarily order the $m$ leaves of $G_\cM$ as $v_1, \ldots, v_m$. Then we inductively define

\begin{equation}
\label{eq:tree-Aw}
A_{w} = \begin{cases}
\{a_i, b_i\} &\text{if $w = v_i$ for $i \in [m]$ where $e_i = \{a_i, b_i\}$,} \\
\bigcup_{v\in D(w)} A_v &\text{otherwise.}
\end{cases}
\end{equation}

We remark that $\sum_{v}|A_v| = O(m \log m)$ since each vertex $i\in V_G$ is present in exactly $\deg (i)$ leaves and in the paths of length $O(\log m)$ that connect those leaves to $r$. We proceed to define a family of $m$-tiles $\cP$ by setting $P_v = \{(i,i) \mid i \in A_v\}$ for every vertex $v$ of the tree, and keeping all the other tiles empty.

Let $\cF$ be a consistent orientation obtained from the application of 
Lemma~\ref{lem-acyclictrans} on $\cM$ and let $\pi$ be the $\cF$-assembly of $\cP$. Since every tile 
is an increasing point set, it follows that $\pi$ belongs to $\Grid(\cM)$.

In order to simplify the rest of the proof, we color $S$ with $n$ colors. We assign a color $i\in V_G$ to a point $p \in S$ with preimage $(i,i)$ in $P_{x,y}$. We claim that $S$ satisfies the following conditions:
\begin{enumerate}[(a)]%[noitemsep,label=(\emph{\alph*})]
\item \label{cond-tree1} The subgraph of $G_\pi$ induced by a single color is connected;
\item \label{cond-tree2} For each edge $e_i = \{a_i,b_i\}$ of $G$ there is an edge in $G_\pi$ between a vertex of color $a_i$ and a vertex of color $b_i$.
\end{enumerate}

Fix a color $i \in V_G$. Let $Q_i$ be the set of all vertices $v$ of $G_\cM$ such that $i \in A_v$. Clearly, $Q_i$ induces a connected subtree of~$G_\cM$. Recall that every point of color $i$ has always the coordinates $(i,i)$ inside any tile.
It follows that for points $(i,i)$ in two neighboring tiles, the $\cF$-assembly of $\cP$ transforms them first to points that share one coordinate and then by rotating slightly clockwise makes them either horizontal or vertical neighbors.
Therefore, the subgraph of $G_\pi$ induced by color $i$ is connected, which proves~(\ref{cond-tree1}).

Every leaf $v_{i}$ must be the only non-empty vertex in its row or column. Let us assume the latter case as the other one is symmetric.
Therefore, the two points contained in the image of $P_{v_i}$ form an edge in $G_\pi$ since no other point lies in the vertical strip between them. In particular, the leaf $v_i$ satisfies the condition~(\ref{cond-tree2}) for edge $e_i$.

The conditions (\ref{cond-tree1}) and (\ref{cond-tree2}) together imply that we can obtain a supergraph of $G$ by contracting every monochromatic subgraph of $G_\pi$ to a single vertex and thus, $G$ is a minor of~$G_\pi$. Observe that the total size of $\pi$ is equal to $\sum_{v}|A_v|$ which we showed to be $O(m \log m)$. And since there exist graphs on $n$ vertices with $O(n)$ edges and tree-width $\Omega(n)$, we deduce that $\tw_C(n) \ge \tw_{\Grid(\cM)}(n) \in \Omega(n / \log n)$.
\end{proof}

%All examples of classes of near-linear width we know about have the deep tree property. It might 
%therefore be useful to check the deep tree property of the principal permutation classes for 
%which Berendsohn cannot determine whether they have near-linear width.\mo{This paragraph should perhaps be rephrased?}

We continue by introducing a different property that implies the deep tree property and is at the same time easier to show for a specific class $\cC$.
A permutation class $\cC$ has the \emph{bicycle property} if it contains a monotone grid subclass 
whose cell graph is connected and contains at least two cycles.
%We include the full, rather technical, proof in the full version, providing here with only a brief sketch.
%We include the full, rather technical, proof in Appendix~\ref{apx:proof-bicycle}, providing here with only a brief sketch.

\begin{proposition}\label{pro-bicycle}
If a permutation class $\cC$ has the bicycle property, then it also has the DTP (and therefore near-linear width).
\end{proposition}

%\begin{proof}[Proof idea]
% The proof consists of two parts.
% First, we will show that there is always a grid subclass of $\cC$ that contains in its cell graph two connected cycles of a certain special type.
% To that end, observe that the cell graph $G_\cM$ can either be two cycles connected by a path or one cycle with a chord.
% For the latter case, we show that by replacing each entry in $\cM$ with a suitable $3 \times 3$ matrix, we obtain a matrix $\cN$ such that $\Grid(\cN)$ is a subclass of $\Grid(\cM)$ and moreover, the cell graph $G_\cN$ contains two cycles joined with a path.
% See Figure~\ref{fig:cycle-chord-lite}.
%  
% \begin{figure}
%   \centering
%   \raisebox{-0.5\height}{\includegraphics[width=0.4\textwidth,page=1]{cycle-with-chord-lite}}
%   \hspace{0.2in}
%   \raisebox{-0.5\height}{\includegraphics[width=0.4\textwidth,page=2]{cycle-with-chord-lite}}
%   \caption{Left: a gridding matrix $\cM$ whose cell graph is a cycle with a chord.
%     Right: a gridding matrix $\cN$ such that $\Grid(\cN)$ is contained in $\Grid(\cM)$ and the cell graph $G_\cN$ consists of two cycles joined by a path.
%   }
%   \label{fig:cycle-chord-lite}
% \end{figure}
% 
% In the second step, we find a way to wind a $c$-subdivided binary tree of arbitrary depth into the two cycles joined with a path and thus, showing that $\cC$ has the DTP.
%\end{proof}
Before we prove the proposition, we need to introduce some more terminology. We say that a 
monotone gridding matrix $\cM$ is \emph{increasing} if $\Grid(\cM)=\;\Inc$, it is \emph{decreasing} 
if $\Grid(\cM)=\;\Dec$, and it is \emph{empty} if $\Grid(\cM)=\emptyset$.

We shall now introduce a variant of the $\cF$-assembly acting on matrices.
Suppose that $\cM$ is a $k \times \ell$ monotone gridding matrix with a consistent orientation~$\cF$.
Let $\Lambda$ be a $k \times \ell$ family of monotone gridding matrices $\{\cL^{i,j} \mid i \in [k], 
j \in[\ell] \}$ such that for each $i \in [k]$ and $j \in[\ell]$, $\cL^{i,j}$ is 
an $m \times m$ monotone gridding matrix. Moreover, $\cL^{i,j}$ is empty if $\cM_{i,j}=\emptyset$, 
and $\cL^{i,j}$ is increasing if $\cM_{i,j}\in\{\Inc,\Dec\}$.

Note that we use upper indices to distinguish individual gridding matrices in $\Lambda$, and bottom 
indices to distinguish cells of a given matrix.
A \emph{(matrix) $\cF$-assembly of $\Lambda$} is the $mk \times m \ell$ gridding matrix $\cN$ 
defined as \[\cN_{(i-1) \cdot m + a, (j-1) \cdot m + b} = \Phi_{i}(\Psi_{j}(\cL^{i,j}))_{a,b}\]
where $\Phi_i$ is an identity if $f_c(i) = 1$ and reversal otherwise, while $\Psi_j$ is an identity 
if $f_r(j) = 1$ and complement otherwise.
Observe that we just apply on each matrix the same symmetry transformation as in the regular 
$\cF$-assembly, and that $\Grid(\cN)$ is a subclass of $\Grid(\cM)$.

\begin{proof}[Proof of Proposition~\ref{pro-bicycle}]
  The proof consists of two parts.
  First, we will show that there is always a grid subclass of $\cC$ that contains in its cell graph two connected cycles of a certain special type.
  And later, we deduce the deep tree property of this particular subclass.
 Let $\cM$ be a monotone gridding matrix such that $\Grid(\cM)$ is a subclass of $\cC$, the cell 
graph $G_\cM$ contains two connected cycles, and replacing any nonempty cell of $\cM$ with 
$\emptyset$ would violate this property. Note that this means that every vertex of $G_\cM$ of degree 
two corresponds to a \emph{corner}, i.e., it has one neighbor in the same row and the other in the 
same column.
  
  We start by showing that we can always find such $\cM$ that furthermore has a consistent orientation $\cF$.
  In order to see that, consider the matrix $\cM^{\times 2}$ obtained from $\cM$ by replacing every $\Inc$-entry with the $2\times 2$ matrix $\begin{psmallmatrix} \cdot  & \Inc \\ \Inc & \cdot  \end{psmallmatrix}$, every $\Dec$-entry with the matrix $\begin{psmallmatrix} \Dec & \cdot \\ \cdot  & \Dec \end{psmallmatrix}$ and every empty entry with the empty $2 \times 2 $ matrix.
  Clearly, $\Grid(\cM^{\times 2})$ is contained in $\Grid(\cM)$.
  Moreover, notice that the orientation $\cF = (f_c, f_r)$ defined as $f_c(i) = (-1)^i$ and $f_r(j) = (-1)^j$ is a consistent orientation for $\cM^{\times 2}$.
  This fact has been first observed by Albert et al.~\cite{Albert2013}.
  Finally, we claim that while the cell graph $G_{\cM^{\times 2}}$ might not be connected, its every connected component in fact contains at least two cycles.
  This follows since both of the new vertices created by replacing a single non-empty entry with a $2 \times 2$ matrix inherit the degree of the original vertex in $G_\cM$.
  There are no vertices of degree one due to our choice of $\cM$ as a minimal matrix with the property and moreover, there is at least one vertex of degree at least 3.
  Thus, we can find two connected cycles by walking from a single vertex in three different directions.
  
  From now on, we assume that $\cM$ has a consistent orientation $\cF$.
  Its cell graph $G_\cM$ can either consist of two cycles connected by a path, or one cycle with a 
(possibly subdivided) chord.
  Let us show that, in fact, we can without loss of generality assume that the former holds.
  Note that we consider the degenerate case of two cycles joined via a vertex of degree four as two cycles connected by a path of zero length.

  To that end, assume that $G_\cM$ forms a cycle with a chord and let $u$ and $v$ be the two vertices with degree 3.
  We call the shortest of the three paths between $u$ and $v$ the \emph{chord of $G_\cM$} and the other two parts the \emph{arcs of $G_\cM$}.
  We orient the arcs such that they together form an oriented cycle.
%  A \emph{corner} is any vertex in $G_\cM$ that has degree 2 and shares neither a common row nor a 
% common column with its both neighbors.
  Since the chord is the shortest of the three parts, there exists at least one corner on each arc 
-- let us choose arbitrarily one corner from each arc, and denote them by $x$ and $y$.
  
  We define a family of matrices $\Lambda$ where each $\cL^{i,j}$ is an increasing $3\times 3$ matrix
  that is a symmetry of one of the four following types
  \[P = \begin{psmallmatrix} \cdot  & \cdot & \cdot \\ \cdot  & \Inc & \cdot \\ \Inc & \cdot & \cdot  \end{psmallmatrix},\quad
  Q = \begin{psmallmatrix} \cdot  & \cdot & \Inc \\ \cdot  & \Inc & \cdot \\ \cdot & \cdot & \cdot  \end{psmallmatrix},\quad
  R = \begin{psmallmatrix} \cdot  & \Inc & \cdot \\ \Inc  & \cdot & \cdot \\ \cdot & \cdot & \cdot  \end{psmallmatrix},\quad
  S = \begin{psmallmatrix} \cdot  & \cdot & \cdot \\ \cdot  & \Inc & \cdot \\ \cdot & \cdot & \cdot  
\end{psmallmatrix}, \]
  where the dots denote the empty entries.
  In particular, for every vertex $z$ on the oriented path from $x$ to $y$ (excluding the endpoints), we set $\cL^z = P$.
  For every vertex $z$ on the oriented path from $y$ to $x$ (again excluding the endpoints), we set $\cL^z = Q$.
  For every vertex $z$ on the chord (again excluding $x$ and $y$), we set $\cL^z = S$.
  We set $\cL^x = R$ if $x$ and its predecessor on the cycle share a common row, otherwise we set $\cL^x = R^{-1}$.
  And finally, $\cL^y = R^{-1}$ if $y$ and its predecessor on the cycle share a common row, otherwise $\cL^y = R$.
  Let $\cN$ be the matrix $\cF$-assembly of $\Lambda$.
  
  Observe that when we limit our view only to the arcs in $\cN$, the non-empty entries actually form two disjoint cycles.
  One cycle contains what was originally the  $(1,1)$ cell of $\cL^z$ for every $z$ on the path from $x$ to $y$ and the $(2,2)$ cell of $\cL^z$ for every $z$ on the path from $y$ to $x$.
  The other cycle contains the $(2,2)$ cells on the path from $x$ and $y$ and the $(3,3)$ cells  on the path from $y$ to $x$.
  The matrices $\cL^x$ and $\cL^y$ act as switches.
  Most importantly, every $\cL^z$ for $z$ in the chord has a single non-empty entry in the $(2,2)$ cell which establishes the connection between the two cycles.
  See Figure~\ref{fig:cycle-chord}.
  
  \begin{figure}
    \centering
    \raisebox{-0.5\height}{\includegraphics[width=0.45\textwidth,page=1]{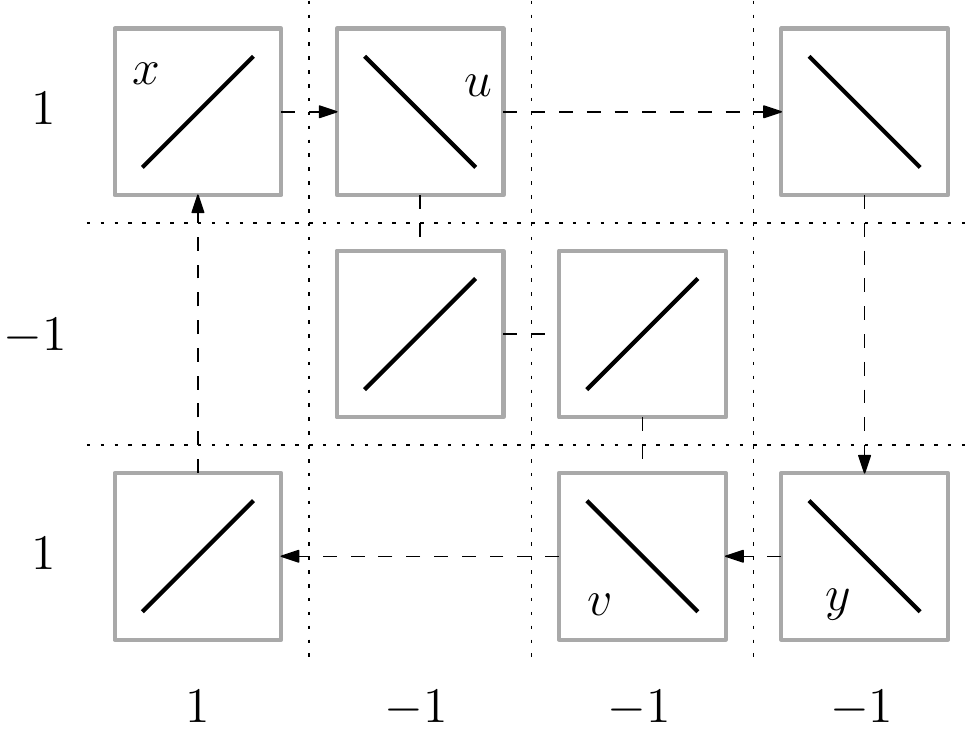}}
    \hspace{0.2in}
    \raisebox{-0.5\height}{\includegraphics[width=0.45\textwidth,page=2]{cycle-with-chord}}
    \caption{Left: a gridding matrix $\cM$ whose cell graph is a cycle with a chord together with a consistent orientation $\cF$ given by the numbers along the edges.
      Right: a gridding matrix $\cN$ whose grid class is contained in $\Grid(\cM)$ and whose cell graph consists of two cycles joined by a path. The numbers along the edges are the coordinates of the original cells under matrix $\cF$-assembly.
    }
    \label{fig:cycle-chord}
  \end{figure}
  
  Henceforth, we assume without loss of generality that $G_\cM$ contains two cycles joined with a path.
  Again, let $u$ and $v$ be the vertices with degree 3 and we orient each cycle arbitrarily.
  We pick one corner on each cycle, let $x$ be a corner on the cycle incident with $u$ and $y$ a corner on the cycle incident with $v$.
  These corners shall again play the role of some sort of switches.
  
  Let $T$ be a complete binary rooted tree of depth $d$ and let $m = 2^{d}-1$ denote the number of vertices in $T$.
  We identify the vertices of $T$ with the set $[m]$ using a breadth-first search numbering.
  Formally, the root is the integer 1 and the remaining vertices are numbered by layers and in each layer from left to right.
  Let $L(i)$ and $R(i)$ denote the left and the right child of a vertex $i$ in $T$.
  
  We now define a family of matrices $\Lambda$ that will create a subdivided $T$ upon its $\cF$-assembly (see Figure~\ref{fig:deep-tree}).
  Each $\cL^{i,j}$ is the following $2m \times 2m$ matrix.
  \begin{itemize}
    \item For any vertex $z$ on the path from $v$ to $x$ (including both $u$ and $v$), we define $\cL^z$ as the whole identity matrix with $2m$ non-empty entries on its diagonal.
    \item For $z$ on the path from $v$ to $y$ (excluding $v$ and $y$), we define $\cL^z$ as the matrix with non-empty entries of the form $(2i - 1, 2i - 1)$ for $i \in [m]$.
    \item For $z$ on the path from $y$ to $v$ (again without its endpoints), we define $\cL^z$ as the matrix with non-empty entries of the form $(2i, 2i)$ for $i \in [m]$.
    \item For the corner $y$, we set $\cL^y$ to contain the non-empty entries $(2i-1, 2i)$ if $y$ and its predecessor on the cycle share a common row, otherwise its non-empty entries are of the form $(2i, 2i-1)$ for $i \in [m]$.
    \item For $z$ on the path from $x$ to $u$ (excluding $x$ and $u$), we define $\cL^z$ as the matrix with non-empty entries of the form $(2i-1, 2i-1)$ for $i \in [m]$.
    \item Finally for the corner $x$, we set $\cL^x$ to contain the non-empty entries $(2 \cdot L(i) - 1, 2i-1)$ and $(2 \cdot R(i) - 1, 2i)$ for every $i$ that is not a leaf in $T$ if $x$ and its predecessor on the cycle share a common row, otherwise we swap the $x$- and $y$-coordinates.
  \end{itemize}
  
  \begin{figure}
    \centering
    \raisebox{-0.5\height}{\includegraphics[width=0.2\textwidth]{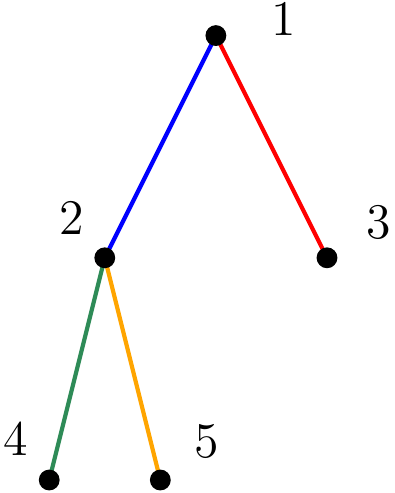}}
    \hspace{0.08\textwidth}
    \raisebox{-0.5\height}{\includegraphics[width=0.6\textwidth]{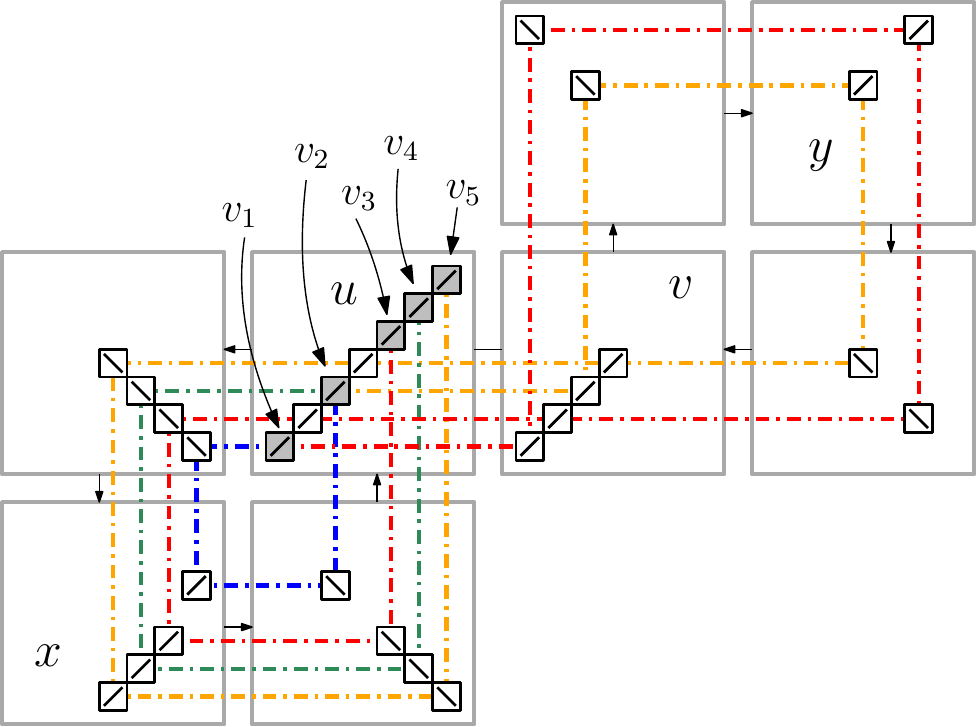}}
    \caption{Embedding a subdivision of the tree $T$ on the left in a grid class whose cell graph consists of two cycles joined by a path, on the right. The grey cells in the submatrix obtained by replacing $u$ denote precisely the vertices corresponding to vertices of $T$, the subdivided edges are highlighted using matching colors.
      The cells not participating in the subdivision of $T$ are omitted.}
    \label{fig:deep-tree}
  \end{figure}
  
  First, we need to verify that every matrix is increasing.
  Luckily, that is obvious in all cases except for $\cL^x$ and in that case it follows from the chosen labeling of the vertices in $T$.
  Let $\cN$ be the $\cF$-assembly of $\Lambda$ and let $v_i$ for $i \in [m]$ be the vertex of $G_\cN$  originating from the $(2i-1, 2i-1)$ cell in $\cL^u$.
  We claim that for every $i$ that is not a leaf in $T$, there are paths of constant length from $v_i$ to both $v_{L(i)}$ and $v_{R(i)}$ and moreover, all these paths are disjoint.
  It then follows that $G_\cN$ contains the desired subdivision of $T$.
  
  Let us start with the vertex $v_i$ for an inner vertex $i$ of $T$.
  Following the path from $u$ to $x$ in $G_\cM$, we inductively see that $v_i$ is connected to the vertices originating from the $(2i-1, 2i-1)$ cells.
  The submatrix corresponding to the original vertex $x$ then acts as a switch using the $(2 \cdot L(i) - 1, 2i-1)$ cell and we can extend the path to all the vertices originating from the $(2 \cdot L(i) - 1, 2 \cdot L(i)-1)$ cells on the path from $x$ back to $u$.
  Therefore, we found a path of constant length that connects $v_i$ to $v_{L(i)}$.
  
  In the case of the right child, we first follow the path from $u$ to $v$ and subsequently to $y$ using  in every step the vertex arising from the $(2i-1, 2i-1)$ cell.
  The corner $y$ again acts as a switch via the $(2i-1, 2i)$ cell and we extend the path all the way back through $v$ to $u$ using the $(2i, 2i)$ cells.
  From there, we follow the path to $x$ and then back to $u$.
  Using the same arguments as for the left child, we see that we end up in the vertex $v_{R(i)}$.
  It is easy to see that the paths corresponding to different edges are indeed pairwise disjoint and thus, we verified the deep tree property.
\end{proof}

For monotone grid classes, the results of this section imply a sharp dichotomy.

\begin{corollary}
  \label{cor:grid-dichotomy}
  For a monotone grid class $\Grid(\cM)$ exactly one of the following holds.
  \begin{itemize}
    \item $G_\cM$ is acyclic and $\tw_\cC(k) \in \Theta(1)$.
    \item $G_\cM$ contains at most one cycle in each component, $\cC$ has LPP and $\tw_\cC(k) \in \Theta(\sqrt k)$.
    \item $G_\cM$ has a component with at least two cycles, $\cC$ has DTP and $\tw_\cC(k) \in \Omega(k/\log k)$.
  \end{itemize}
  
\end{corollary}

\subsection{The case of principal classes}\label{ssec-principal}

%TODO rewrite to talk about LPP and DTP
In this section, we investigate the long path and deep tree properties of principal classes, i.e., the classes of the form $\Av(\pi)$.
Combined with the results of Subsection~\ref{subsec:width}, it allows us to infer lower bounds for the tree-width growth function of $\Av(\sigma)$.
Moreover, together with the results of Section~\ref{sec-hardness}, we obtain conditional lower 
bounds for counting patterns from $\Av(\sigma)$.
% by proving that they contain or conversely cannot contain certain  monotone grid subclasses.
We note that Berendsohn~\cite{BerendsohnMs} has previously shown that for any $\pi$ of length 
at least 4 that is not symmetric to one of $\{3412,\allowbreak 3142,\allowbreak 4213,\allowbreak 
4123,\allowbreak 42153,\allowbreak 41352,\allowbreak 42513\}$, the class $\Av(\pi)$ has near-linear 
width.
We reproduce and extend this result in a concise way with the tools that we have built up.

The \emph{$k$-step increasing $(\cC, \cD)$-staircase}, denoted by $\St_k(\cC, \cD)$ is a grid class $\Grid(\cM)$ of a $k \times (k+1)$ gridding matrix $\cM$ such that the only non-empty entries in $\cM$ are $\cM_{i,i} = \cC$ and $\cM_{i, i+1} = \cD$ for every $i \in [k]$. In other words, the entries on the main diagonal are equal to $\cC$ and the entries of the adjacent lower diagonal are equal to $\cD$. The \emph{increasing $(\cC, \cD)$-staircase}, denoted by $\St(\cC, \cD)$, is the union of $\St_k(\cC, \cD)$ over all $k \in \mathbb{N}$.

The authors~\cite{Jelinek2021} recently showed that $\Av(\sigma)$ contains a certain staircase class for three patterns of length 3 and certain $2\times 2$ grid classes for four patterns of length 4.

\begin{proposition}[Jelínek et al.\cite{Jelinek2021}]
  \label{prop:staircases}
  We have $\St(\Inc,\Av(321))\subseteq \Av(4321)$, $\St(\Inc,\Av(231))\subseteq\Av(4231)$ and $\St(\Inc,\Av(312))\subseteq \Av(4312)$.
\end{proposition}

\begin{proposition}[Jelínek et al.\cite{Jelinek2021}]
  \label{prop:rich-2by2}
  The class $\Av(\sigma)$ contains the class $\Grid(\cM)$ for the gridding matrix $\cM = \begin{psmallmatrix} \Dec & \Inc \\ \Av(\pi) & \Dec \end{psmallmatrix}$ whenever
  \begin{itemize}%[noitemsep]
    \item $\pi = 132$ and $\sigma = 14523$, or
    \item $\pi = 231$ and $\sigma = 24513$, or
    \item $\pi = 321$ and $\sigma \in \{32154, 42513\}$.
  \end{itemize}
\end{proposition}

\begin{figure}[t]
\includegraphics[width=\textwidth]{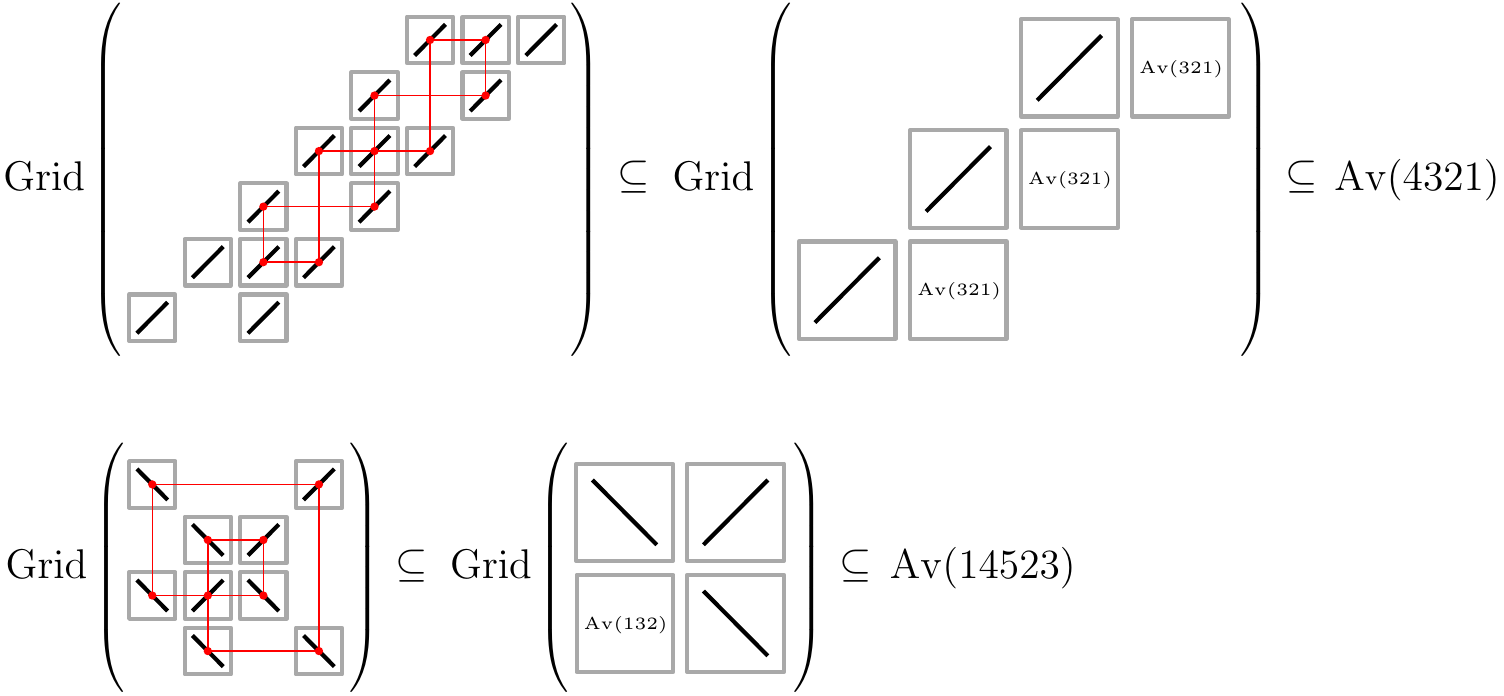}
\caption{Establishing the bicycle property with the help of Propositions~\ref{prop:staircases} (top) 
and~\ref{prop:rich-2by2} (bottom). The red lines highlight the two connected 
cycles in the cell graphs.}\label{fig-bicycle}
\end{figure}

\begin{proposition}
  \label{prop:dtp-principal}
  If $\sigma$ is a permutation of length at least 4 that is not in symmetric to any of $3412, 3142, 4213, 4123$ or $41352$, then $\Av(\sigma)$ has the bicycle property and thus, $\Av(\sigma)$ has near-linear width.
\end{proposition}
\begin{proof}
As shown by the authors~\cite{Jelinek2021}, every $\sigma$ of length at least 4 that is not  
symmetric to any of $3412, 3142, 4213, 4123$ or $41352$ must contain, up to symmetry, one of the 
three permutations $4321$, $4231$ and $4312$ from Proposition~\ref{prop:staircases} or one of the 
four permutations $14523$, $24513$, $32154$ and $42513$ from Proposition~\ref{prop:rich-2by2}. It is 
thus enough to establish the bicycle property for $\Av(\sigma)$ when $\sigma$ is one of these seven 
permutations. Refer to Figure~\ref{fig-bicycle}.

We start by proving that every class defined by forbidding a pattern of length 3 must contain a special type of monotone grid subclass.
For arbitrary $\pi$ of length 3, the class $\Av(\pi)$ contains a grid class $\Grid(\cM)$ such that $\cM$ is a $2\times 2$ monotone gridding matrix with three non-empty entries.
Since there are only two different symmetry types of permutations of length 3, it is enough to check that
\[\Grid\begin{psmallmatrix} \Inc & \Inc \\ \Inc &  \cdot  \end{psmallmatrix}  \subseteq \Av(321)
\quad \text{and} \quad
\Grid\begin{psmallmatrix} \Dec & \Inc \\ \cdot  & \Dec  \end{psmallmatrix}  \subseteq \Av(132).
\]

First, we prove the claim for the patterns that appear in Proposition~\ref{prop:staircases}.
Let $\sigma \in \{4321,\allowbreak 4231, \allowbreak 4312\}$ and take a 3-step increasing staircase $\St_3(\Inc,\Av(\pi))$ for $\pi$ of length 3 that is contained in $\Av(\sigma)$.
Let $\cM'$ be a $6\times8$ monotone gridding matrix obtained from $\St_3(\Inc,\Av(\pi))$ by replacing every $\Inc$-entry by the identity matrix $\begin{psmallmatrix} \cdot  & \Inc \\ \Inc & \cdot  \end{psmallmatrix}$ and every $\Av(\pi)$-entry with its $2 \times 2$ monotone grid subclass which has three non-empty entries.
Clearly, $\Grid(\cM')$ is a subclass of $\Av(\sigma)$, and it is easy to check that for any $\pi$, the cell graph of $\cM'$ is connected and contains two cycles. 
  
We prove the claim for the patterns that appear in Proposition~\ref{prop:rich-2by2} in a similar fashion.
Let $\sigma \in \{14523 \allowbreak, 24513, \allowbreak 32154,\allowbreak 42513\}$ and take $\cM$ to be the grid class $\Grid\begin{psmallmatrix} \Dec & \Inc \\ \Av(\pi) & \Dec \end{psmallmatrix}$ for $\pi$ of length 3 that is contained in $\Av(\sigma)$.
Similar to before, let $\cM'$ be the gridding matrix obtained from $\cM$ by replacing the $\Inc$-entry with the matrix $\begin{psmallmatrix} \cdot  & \Inc \\ \Inc & \cdot  \end{psmallmatrix}$, both $\Dec$-entries with the matrix $\begin{psmallmatrix} \Dec & \cdot \\ \cdot  & \Dec \end{psmallmatrix}$, and $\Av(\pi)$ with its $2 \times 2$ monotone grid subclass which has three non-empty entries.
Again, $\Grid(\cM')$ is a subclass of $\Av(\sigma)$, and it is easy to check that for any $\pi$, the cell graph of $\cM'$ is connected and contains two cycles. 
\end{proof}

%See the full version for the whole discussion.
%We include the whole discussion in Appendix~\ref{sec:principal-without-dtp}.\mo{in the full version}

\begin{figure}[t]
  \begin{tabularx}{\textwidth}{|>{\raggedright}p{70pt}| >{\raggedright}p{80pt}| >{\hsize=1\hsize\arraybackslash}X|}
    \hline
    $\sigma$ & LPP, DTP of $\Av(\sigma)$ & Comment \\ \hline
    1, 21, 312 & neither LPP nor DTP &  $\tw_{\Av(\sigma)} \in \Theta(1)$ by Ahal and 
Rabinovich~\cite{AR08_subpattern} which implies absence of DTP and LPP by 
Proposition~\ref{pro-path}. \\\hline
    321, 3412, 3142, 4213, 4123, 41352 & LPP but not DTP  & LPP of 321 and 3412 follows due to Jelínek and Kynčl~\cite{JeKy}, the rest contains either 123 or 321. 
    The absence of DTP follows from Observation~\ref{obs-3412} and Propositions~\ref{pro-4123} 
and~\ref{pro-41352}. \\\hline
    All other & both LPP and DTP & DTP by Proposition~\ref{prop:dtp-principal}, LPP follows. \\\hline
  \end{tabularx}
  \caption{The long path and deep tree properties of principal classes, i.e. classes of form $\Av(\sigma)$. Only one pattern $\sigma$ from each symmetry group is listed.}
  \label{table:lpp-dtp}
\end{figure}

%\section{Principal classes that cannot have DTP}\label{sec:principal-without-dtp}
We continue by showing that the DTP cannot get us any further, since for any $\sigma\in\{3412,\allowbreak 3142,\allowbreak 
4213,\allowbreak 4123,\allowbreak 41352\}$, the class $\Av(\sigma)$ does not have the DTP.
Hereby, we actually obtained a complete knowledge of LPP and DTP for principal classes.
See Figure~\ref{table:lpp-dtp}.
%The following observations and propositions show that for $\sigma\in\{3412,\allowbreak 3142,\allowbreak 
%4213,\allowbreak 4123,\allowbreak 41352\}$, the class $\Av(\sigma)$ does not have the deep tree 
%property.
Note that $\Av(3142)$ and $\Av(4213)$ are, up to symmetry, subclasses of $\Av(41352)$, so 
we only need to focus on the patterns $3412$, $4123$ and $41352$.

Let us say that a graph $G$ is \emph{representable} in a class $\cC$ if $\cC$ contains a monotone 
grid subclass whose cell graph is~$G$. 

\begin{observation}\label{obs-3412}
  A graph with a vertex of degree three is not representable in $\Av(3412)$.
\end{observation}

\begin{proposition}\label{pro-4123}
  If $G$ is a graph with a vertex of degree 3 whose every neighbor has degree at least 2, then $G$ is 
  not representable in $\Av(4123)$.
\end{proposition}
\begin{proof}
  Suppose $G$ is a graph with a vertex $v$ of degree 3, and let $x$, $y$ and $z$ be the three 
  neighbors of $v$. Suppose each of the three neighbors has degree at least 2. For contradiction, 
  suppose that $\Av(4123)$ contains a monotone grid subclass $\cC=\Grid(\cM)$ whose cell graph is~$G$. 
  
  Abusing notation slightly, we will identify the vertices of $G$ with the corresponding cells of the 
  matrix~$\cM$. Since the pattern $4123$ is symmetric with respect to diagonal reflection, we may 
  assume without loss of generality that $x$ and $y$ are in the same column of $\cM$ as $v$, and that 
  $z$ is in the same row as $v$. Assume further that $x$ is above $v$ and $y$ is below~$v$.
  
  By assumption, $x$ has degree at least 2. In particular, there is a vertex $w$ adjacent to $x$ and 
  different from~$v$. The vertex $w$ is either in the same row as $x$ or in the same column and 
  above~$x$.  However, for any possible placement of $w$, the four cells $v, w, x,y$ will embed the 
  forbidden pattern $4123$, a contradiction.
\end{proof}

Finally, we turn to the pattern $41352$.  Here the proof is slightly more involved and we begin with 
a lemma. Note that we assume that the rows in a gridding matrix are numbered bottom to top. A cell 
in row $r$ and column $c$ of a gridding matrix is referred to as the \emph{cell $(r,c)$}. 

\begin{lemma}\label{lem-41352}
  Let $\cC=\Grid(\cM)$ be a monotone grid class not containing the pattern 41352. Suppose that there 
  are two row indices $r_1<r_2$ and two column indices $c_1<c_2$, such that the three cells 
  $(r_1,c_2)$, $(r_2,c_1)$ and $(r_2,c_2)$ are all nonempty, and moreover the cell $(r_2,c_2)$ is a 
  $\Inc$-cell. Then the following holds:
  \begin{enumerate}
    \item The cell $(r_1,c_2)$ is a $\Dec$-cell.
    \item Any cell $(r,c)$ satisfying $r_1\le r\le r_2$ and $c\ge c_2$ is empty, except the cells 
    $(r_1,c_2)$ and $(r_2,c_2)$.
    \item Any cell $(r,c)$ satisfying $r\le r_1$ and $c_1\le c\le c_2$ is empty, except the cell 
    $(r_1,c_2)$.
  \end{enumerate}
\end{lemma}
\begin{proof}
  If $(r_1,c_2)$ were a $\Inc$-cell, we could embed $41352$ by mapping the values $1,2$ into cell 
  $(r_1,c_2)$, values $3,5$ into $(r_2,c_2)$ and value $4$ into $(r_2,c_1)$. This proves the first 
  claim.
  
  If there were a nonempty cell $(r,c)$ with $r_1\le r\le r_2$ and $c\ge c_2$, and with 
  $(r,c)\not\in\{(r_1,c_2),(r_2,c_2)\}$, we could embed 2 into this cell, 1 into cell $(r_1,c_2)$, $3$ 
  and $5$ into cell $(r_2,c_2)$, and $4$ into $(r_2,c_1)$. This proves the second claim. The third 
  claim is analogous.
\end{proof}

Note that the pattern $41352$ is preserved under rotations by a multiple of~$90^\circ$.  Thus, the 
previous lemma remains valid when the entire gridding matrix $\cM$ is rotated in such a way; note 
that a $90^\circ$-rotation transforms a $\Dec$-cell into a $\Inc$-cell and vice versa.

\begin{proposition}\label{pro-41352}
  Let $G$ be a graph containing a vertex $v$ of degree 3, whose three neighbors all have degree 3. 
  Then no subdivision of $G$ is representable in $\Av(41352)$.
\end{proposition}
\begin{proof}
  For contradiction, suppose we have a monotone grid class $\cC=\Grid(\cM)\subseteq\Av(41352)$ whose 
  cell graph $G'$ is a subdivision of~$G$. Let the vertex $v$ correspond to a cell $(r,c)$ of $\cM$. 
  By rotational symmetry, we may assume that the three neighbors of $v$ in $\cM$ are to the left, 
  to the top and below~$v$.  Let $x_1$ be the neighbor of $v$ situated below~$v$, and let $y$ be the 
  neighbor of $v$ situated to its left.
  
  Note that $v$ must be a $\Inc$-cell, else we could embed values $2,3$ into $v$, and the remaining 
  three values into the three neighbors of~$v$.  
  
  We will show that the connected component of $G'-v$ containing the vertex $x_1$ does not contain any 
  vertex of degree greater than $2$, contradicting the structure of~$G'$. Applying 
  Lemma~\ref{lem-41352} to the three vertices $y$, $v$ and $x_1$, we conclude that $x_1$ is a 
  $\Dec$-cell, and that is has no neighbor to its bottom or to its right. If $x_1$ has degree two, 
  then its neighbor different from $v$ is a cell $x_2$ situated to its left. We may then again apply 
  Lemma~\ref{lem-41352} (rotated $90^\circ$ clockwise) to the three vertices $v$, $x_1$ and $x_2$, 
  concluding that $x_2$ is a $\Inc$-cell, and any potential neighbor of $x_2$ different from $x_1$ is 
  located above $x_2$. Denoting such a neighbor $x_3$, and applying Lemma~\ref{lem-41352} to the 
  vertices $x_1$, $x_2$ and $x_3$, we again conclude that $x_3$ is a $\Dec$-cell. 
  
  Continuing in this fashion, we may inductively show that the connected component of $G'-v$ containing 
  the vertex $x_1$ is a path $x_1, x_2, x_3,\dotsc$ arranged into a clockwise spiral and consisting of 
  an alternation of $\Inc$-cells and $\Dec$-cells. In particular, the component does not contain any 
  vertex of degree 3, contradicting the structure of~$G'$.
\end{proof}

\section{Hardness of \#PPM}\label{sec-hardness}

In this section, we provide conditional lower bounds for modified variants of \PPPM{$\cC$} given LPP or DTP.
The results of this section are proved under a slightly stronger assumptions about the classes.
Apart from the LPP or DTP property, we furthermore require an algorithm that provides a witnessing long path or deep tree.
Formally, a class $\cC$ has the \emph{computable LPP} if it has the LPP and there is an algorithm that, for a given $k$, outputs the description of a monotone grid subclass of $\cC$ whose cell graph is a path of length $k$.
Similarly, a class $\cC$ has the \emph{computable DTP} if it has the DTP and there is a constant 
$c$ and an algorithm that, for a given $k$, outputs the description of a monotone grid subclass of 
$\cC$ whose cell graph is a $c$-subdivided binary tree of depth~$k$.
Observe that all the specific examples of classes we encountered (and especially the principal 
classes in Subsection~\ref{ssec-principal}) possess the computable version of their corresponding 
properties.

We will reduce from the well-known problem \emph{partitioned subgraph isomorphism} (\PSIx) defined as follows.
We receive on input two graphs $G = (V_G, E_G)$ and $H = (V_H, E_H)$ together with a coloring $\chi\colon V_H \to V_G$ of vertices of $H$, using the vertices of $G$ as colors.
We have to decide if there is  a mapping $\phi\colon V_G \to V_H$ such that whenever $\{u,v\} \in E_G$ then also $\{\phi(u), \phi(v)\} \in E_H$ and moreover $\chi(\phi(v)) = v$ for every $v \in V_G$.
Less formally, we aim to find $G$ as a subgraph of $H$, but we prescribe in advance where each vertex 
can be mapped.
It is a well-known fact that \PSI is hard to solve.

\begin{theorem}[Marx~\cite{Marx2010}, Bringmann et al.~\cite{Bringmann2016}]
  \label{thm:psi-hardness}
  Unless ETH fails, \PSI cannot be solved
  in time $f(k) \cdot n^{o(k/ \log k)}$
  for any function $f$, where $n = |V_H|$ and $k = |E_G|$. This is true even when we require $G$ to have exactly as many vertices as edges.
\end{theorem}

If we additionally fix $G$ to be the clique on $k$ vertices we obtain the problem called \textsc{Partitioned Clique}.
Formally, the input to \textsc{Partitioned Clique} consists of a graph $H = (V_H, E_H)$ together with a coloring $\chi\colon V_H \to [k]$ and we have to decide if there is a $k$-clique in $H$ that hits all $k$ available colors.
It is easy to see that \textsc{Partitioned Clique} can be solved in time $f(k) \cdot n^{O(k)}$.
However, there is also a matching conditional lower bound.

\begin{theorem}[Cygan et al.~\cite{Cygan2015}]
  \label{thm:clique-hardness}
  Unless ETH fails, \textsc{Partitioned Clique} cannot be solved
  in time $f(k) \cdot n^{o(k)}$
  for any function $f$, where $n = |V_H|$.
\end{theorem}

We shall also not reduce directly to the problems of interest.
Rather, we first reduce to the \emph{$\cC$-Pattern Anchored PPM} (\APPPM{$\cC$}) problem, defined as 
follows. The input consists of permutations $\pi \in \cC$ and arbitrary $\tau$ together with pairs of 
points $A$ in $\pi$ and $B$ in $\tau$ that are called \emph{anchors}. Moreover, we are promised that 
if the two points of $A$ form an increasing pair, then we may inflate each of the two points in $A$ 
into an increasing sequence of arbitrary length with the resulting permutation still belonging to 
$\cC$, and symmetrically, if $A$ is a decreasing pair, then the two points can be inflated to 
arbitrary decreasing sequences.
The goal of \APPPM{$\cC$} is to decide whether there is an embedding of $\pi$ into $\tau$ that maps 
$A$ to~$B$.

For $\cC$ with the computable LPP, we are able to reduce \textsc{Partitioned Clique} to \APPPM{$\cC$} such that the size of the pattern $\pi$ is linear in the number of vertices of the clique. 
And for $\cC$ with the computable DTP, we provide a reduction from \PSI to \APPPM{$\cC$} such that the size of $\pi$ is almost linear in the size of the graph $G$.
%Due to the space constraints and technicality of the reductions, we include here only brief overviews and describe both of them, including the proofs of correctness in the full version.

\subsection{Building blocks}
\label{subsec:building-blocks}

Before diving into either of the two reductions, we set up some tools used by both of them.
In particular, we define several types of tiles.
All tiles will be constructed from pairs of points called \emph{atomic pairs} that form the pattern 12.
We will guarantee that an atomic pair in the pattern must map to a single atomic pair in the text which justifies their name.

Fix an instance $(G,H, \chi)$ of \PSI and we set $n = |V_H|$, $k = |V_G|$.
Let us identify the vertices of $V_G$ with the set $[k]$ and let $V_a \subseteq V_H$ be the set of all vertices colored by $a \in [k]$.
Notice that $V_1, \dots, V_k$ form a partition of the set $V_H$.
Moreover, set $n_a = |V_a|$ and let us choose an arbitrary order of vertices in $V_a$ denoting them $v^a_i$ for $i \in [n_a]$.
To every vertex $v^a_i$, we associate two values -- the \emph{rank of $v^a_i$} denoted by $\alpha^a_i$ and the \emph{reverse rank of $v^a_i$} denoted by $\beta^a_i$ where
\[\alpha^a_i = \sum_{b < a} n_b + i-1 \quad \text{ and } \quad \beta^a_i = \sum_{b < a} n_b + n_a -i.\]
Observe that the rank corresponds to the lexicographic order of $v^a_i$ by $(a,i)$ and the reverse rank corresponds to the lexicographic order by $(a, n_a - i)$.

Set $m = 3n$.
Let $E_a$ be the set of all pairs $(v^a_i,v^b_j)$ such that either the edge $\{a, i\} \not \in E_G$ or $\{v^a_i,v^b_j\} \in E_H$.
We shall build our reductions from only a handful different types of tiles that we now proceed to describe.

%\subparagraph{Anchor tile.}
An \emph{anchor tile}, denoted by \anchortile, contains a single atomic pair
\begin{equation*}
  \label{eq:anchor-tile}
  T = \left\{(0,0),\; ((n+1) \cdot m, (n+1) \cdot m)\right\}.
\end{equation*}

%The lower left and upper right corner of each tile contains atomic pair, called \emph{guards}, whose purpose is to force that any embedding must map the image of $P_i$ to the image of $T_i$.
%In particular, every tile $P_i$ for $i \ge 2$ contains  atomic pairs 
%\[(1,1), (2,2) \quad \text{and} \quad  (2k+3, 2k+3), (2k+4, 2k+4).\]
%And every tile $T_i$ for $i \ge 2$ contains atomic pairs 
%\[(1,1), (2,2) \quad \text{and} \quad  (nm+3, nm+3), (nm+4, nm+4).\]

\subparagraph{Pattern tile.}
%Every other tile $P_i$ for $i \ge 2$ is in fact identical and consists of $k$ atomic pairs $X^i_1, \dots, X^i_k$ in between, one per each vertex of $G$, where
For every $W \subseteq [k]$, a \emph{pattern tile of $W$}, denoted by \patterntile{$W$}, contains for each $a \in W$ an atomic pair
\begin{equation*}
  \label{eq:pattern-tile}
  X_a = \{(2a+1, 2a+1),\; (2a+2, 2a+2)\}
\end{equation*}
Observe that any pattern tile is simply an increasing sequence of length $2|W|$.

%We call the tile $T_2$, an \emph{assignment tile} as it serves the purpose of simulating a mapping $\phi\colon V_G \to V_H$ that respects the colorings $\chi$. It will contain for each $i \in [k]$ and $j \in [n_i]$ an atomic pair
%\begin{gather}
  %  \label{eq:pair-Y2}
  %  \left(\alpha^i_j \cdot m + 1,\; \beta^i_j \cdot m + 1\right) \quad \text{ and } \quad 
  %  \left(\left(\alpha^i_j + 1\right) \cdot m ,\; \left(\beta^i_j + 1\right) \cdot m\right).
  %\end{gather}

%Observe that each pair $Y^2_{i,j}$ forms an occurrence of 12 and the atomic pairs corresponding to the vertices of $V_i$ form skew sum of $n_i$ copies of 12.
%On the other hand for any $i < j$, the atomic pairs of $V_i$ lie all in a block to the left and below of all the atomic pairs of $V_j$.\mo{Figure}

\begin{figure}
  \centering
  \hspace{-0.6in}
  \raisebox{-0.5\height}{\includegraphics[width=0.5\textwidth,page=1]{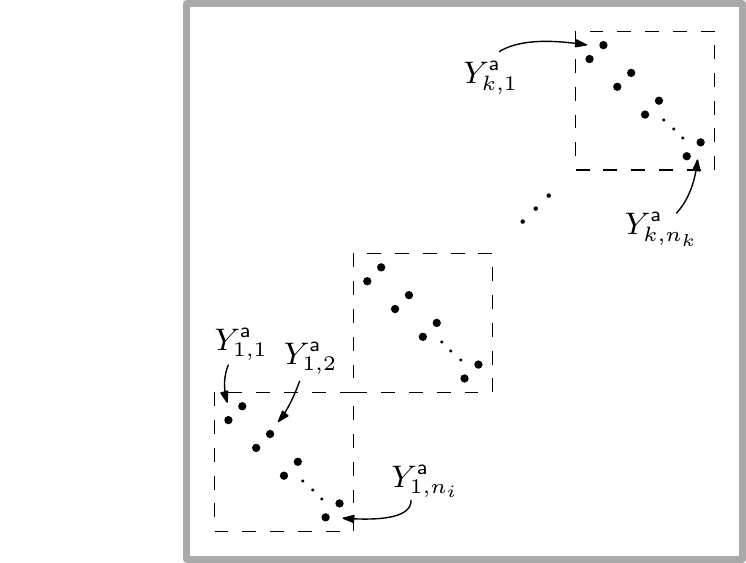}}
  \hspace{0.2in}
  \raisebox{-0.5\height}{\includegraphics[width=0.5\textwidth, page=2]{tiles}}
  \caption{An assignment tile $\assigntile$ on the left and a branch tile $\branchtile{a}{W}$ on the right. Note that the branch tile is zoomed to the vertical strip between $\alpha^a_i \cdot m$ and $(\alpha^a_{n_a}+1) \cdot m$ as the rest of the matrix is empty.}
  \label{fig:tiles}
\end{figure}

\subparagraph{Assignment tile.}
An \emph{assignment tile}, denoted by \assigntile, contains for each $a \in [k]$ and $i \in [n_a]$ an atomic pair $Y^\msfa_{a,i}$ consisting of points
\begin{gather*}
  \label{eq:pair-Y2}
  Y^\msfa_{a,i} = \left\{
  \left(\alpha^a_i \cdot m + 3,\; \beta^a_i \cdot m + 3\right),\;
  \left(\left(\alpha^a_i + 1\right) \cdot m + 2 ,\; \left(\beta^a_i + 1\right) \cdot m + 2\right)\right\}.
\end{gather*}

Observe that each pair $Y^\msfa_{a,i}$ forms an occurrence of 12 and the atomic pairs corresponding to the vertices of $V_a$ form a skew sum of $n_a$ copies of 12.
On the other hand for any $a < b$, the atomic pairs of $V_a$ lie all in a block to the left and below of all the atomic pairs of $V_b$. See the left part of Figure~\ref{fig:tiles}.
Observe that the assignment tile contains exactly $2 \cdot |V_H|$ points.

\subparagraph{Identity tile.}
For a subset $W \subseteq V_G = [k]$, an identity tile of $W$, denoted by \idtile{$W$}, contains for every $a \in W$ and $j \in [n_a]$ an atomic pair
\begin{gather*}
  \label{eq:identity-tile}
  Y^\msfid_{a,j} = \left\{\left(\alpha^a_j \cdot m + 3,\; \alpha^a_j \cdot m + 3\right), \;
  \left(\left(\alpha^a_j + 1\right) \cdot m + 2, \; \left(\alpha^a_j + 1\right) \cdot m + 2\right)\right\}.
\end{gather*}
Observe that \idtile{$W$} forms an increasing sequence of length $2 \cdot \sum_{a \in W} |V_a|$.

%We call the tile $T_{4a-1}$ an \emph{identity tile}. It consists of an atomic pairs $Y^{4a-1}_{i,j}$ for $i$ in $[k]$ and $j \in [n_i]$ containing points
%\begin{gather}
  %  \label{eq:identity-tile}
  %  \left(\alpha^i_j \cdot m + 1,\; \alpha^i_j \cdot m + 1\right) \quad \text{ and } \quad 
  %  \left(\left(\alpha^i_j + 1\right) \cdot m, \; \left(\alpha^i_j + 1\right) \cdot m\right).
  %\end{gather}

\subparagraph{Branch tile.}
Let $a \in [k]$ and $W \subseteq [k]$ such that $a < b$ for every $b \in W$.
A branch tile of $a$ and $W$, denoted by \branchtile{$a$}{$W$}, contains for every $i \in [n_a]$ an atomic pair
\begin{equation*}
  \label{eq:branch-Y}
  Y^\msfb_{a,i} = \left\{\left(\beta^a_i \cdot m + 3,\; \alpha^a_i \cdot m + 3\right), \quad
  \left(\beta^a_i \cdot m + 4,\; \alpha^a_i \cdot m + 4\right)\right\}.
\end{equation*}
Moreover it contains for every $b \in W$ and $(v^a_i,v^b_j) \in E_a$ an atomic pair
\begin{equation*}
  \label{eq:branch-Z}
  Z^\msfb_{a,i,b,j} = \left\{
  \begin{gathered}
    \left(\beta^a_i \cdot m + 2 \alpha^b_{j} + 5 ,\; \alpha^b_{j} \cdot m + 2 \alpha^a_i + 3\right),\\ \quad
    \left(\beta^a_i \cdot m + 2 \alpha^b_{j} + 6,\; \alpha^b_{j} \cdot m + 2 \alpha^a_i + 4\right)
  \end{gathered}\right\}.
\end{equation*}

%The tile $T_{4a}$ is called a \emph{branch tile}. For every $i \in [a-1]$ and every $j \in [n_i]$, it contains an atomic pair $Y^{4a}_{i,j}$ equal to $Y^{4a-1}_{i,j}$.
%Furthermore, it contains for every $b \in [n_a]$ the atomic pair $Y^{4a}_{a,b}$ consisting of points
%\begin{equation}
  %  \label{eq:branch-Y}
  %  \left(\beta^a_b \cdot m + 1,\; \alpha^a_b \cdot m + 1\right) \quad \text{ and } \quad
  %  \left(\beta^a_b \cdot m + 2,\; \alpha^a_b \cdot m + 2\right).
  %\end{equation}
%For every $(v^a_b,v^i_{j}) \in E_a$, we add an atomic pair $Z^{4a}_{b,i,j}$ consisting of points
%\begin{equation}
  %  \label{eq:branch-Z}
  %  \begin{gathered}
    %      \left(\beta^a_b \cdot m + 2 \alpha^i_{j} + 3 ,\; \alpha^i_{j} \cdot m + 2 \alpha^a_b + 1\right), \text{ and}\\
    %    \left(\beta^a_b \cdot m + 2 \alpha^i_{j} + 4,\; \alpha^i_{j} \cdot m + 2 \alpha^a_b + 2\right).
    %  \end{gathered}
  %\end{equation}

Observe that $Z^\msfb_{a,i,b,j}$ lies in the horizontal strip between $y = \alpha^b_{j} \cdot m + 3$ and $y = (\alpha^b_{j} + 1)\cdot m + 2$. If we look at all the atomic pairs in this strip, they are all of the form $Z^\msfb_{a,i',b,j}$ where  $(v^a_{i'},v^b_{j}) \in E_a$ and they form a skew sum of several copies of 12.
Vertically, $Z^\msfb_{a,i,b,j}$ lies in the strip between $x = \beta^a_{i} \cdot m + 5$ and $x = (\beta^a_i + 1)\cdot m + 2$. And if we look at all the atomic pairs in this strip, they are all of the form $Z^\msfb_{a,i,b',j'}$ where $(v^a_i,v^{b'}_{j'}) \in E_a$ for some $b' > a$. Moreover, all the atomic pairs in this strip form an increasing sequence.
See the right part of Figure~\ref{fig:tiles}.

\begin{figure}
  \centering
  \hspace{-0.6in}
  \raisebox{-0.5\height}{\includegraphics[width=0.5\textwidth,page=3]{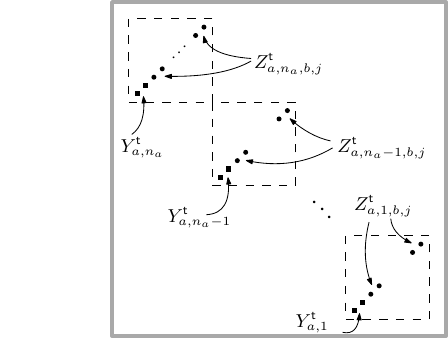}}
  \hspace{0.2in}
  \raisebox{-0.5\height}{\includegraphics[width=0.5\textwidth, page=4]{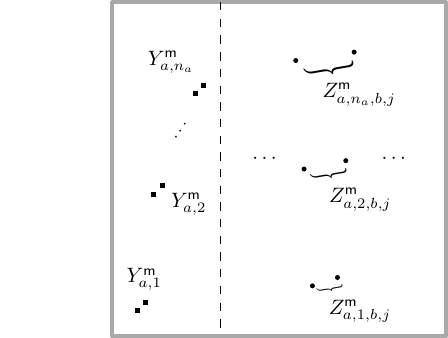}}
  \caption{A test tile $\testtile{a}{W}$ on the left and a branch tile $\mergetile{a}{W}$ on the right. Note that the test tile is zoomed to the box between $\alpha^a_i \cdot m$ and $(\alpha^a_{n_a}+1) \cdot m$ in both axes as the rest of the matrix is empty.
    In the merge tile we focus just on the vertical strip between between $\alpha^a_i \cdot m$ and $(\alpha^a_{n_a}+1) \cdot m$ and on the strip between $\alpha^b_j \cdot m - n$ and $(\alpha^b_{j} + 1) \cdot m + n$.}
  \label{fig:tiles2}
\end{figure}

\subparagraph{Test tile.}
A test tile of $a \in [k]$ and  $W \subseteq [k]$ such that $a < b$ for every $b \in W$, denoted by \testtile{$a$}{$W$}, contains for every $i \in [n_a]$ an atomic pair
\begin{equation*}
  \label{eq:test-Y}
  Y^\msft_{a,i} = \left\{
  \left(\beta^a_i \cdot m + 3,\; \alpha^a_i \cdot m + 3\right), \quad 
  \left(\beta^a_i \cdot m + 4,\; \alpha^a_i \cdot m + 4\right)\right\}.
\end{equation*}

Furthermore for every $b \in W$ and $(v^a_i,v^b_j) \in E_a$, it contains an atomic pair 
\begin{equation*}
  \label{eq:test-Z}
  Z^\msft_{a,i,b,j} = \left\{
  \begin{gathered}
    \left(\beta^a_i \cdot m + 2 \alpha^b_{j} + 5 ,\; \alpha^a_i \cdot m + 2 \alpha^b_{j} + 5 \right),\\
    \left(\beta^a_i \cdot m + 2 \alpha^b_{j} + 6,\; \alpha^a_i \cdot m + 2 \alpha^b_{j} + 6\right)
  \end{gathered}\right\}.
\end{equation*}

The test tile can be obtained as a skew sum of increasing blocks.
For every $i \in [n_a$], there is a block $B_i$ that contains the pair $Y^\msft_{a,i}$ in its bottom left corner followed by pairs $Z^\msft_{a,i,b,j}$ ordered lexicographically  by $(b,j)$.
See the right part of Figure~\ref{fig:tiles2}.

%We call the tile $T_{4a+1}$ a \emph{test tile}.
%For every $i \in [a-1]$ and every $j \in [n_i]$, it contains an atomic pair $Y^{4a+1}_{i,j}$ equal to $Y^{4a-1}_{i,j}$.
%Then it contains for every $b \in [n_a]$ the atomic pair $Y^{4a+1}_{a,b}$ consisting of points
%\begin{equation}
  %  \label{eq:test-Y}
  %  \left(\beta^a_b \cdot m + 1,\; \alpha^a_b \cdot m + 1\right) \quad \text{ and } \quad
  %  \left(\beta^a_b \cdot m + 2,\; \alpha^a_b \cdot m + 2\right).
  %\end{equation}

%For every $(v^a_b,v^i_{j}) \in E_a$, we add an atomic pair $Z^{4a+1}_{b,i,j}$ consisting of points
%\begin{equation}
  %  \label{eq:test-Z}
  %  \begin{gathered}
    %  \left(\beta^a_b \cdot m + 2 \alpha^i_{j} + 3 ,\; \alpha^a_b \cdot m + 2 \alpha^i_{j} + 3 \right), \text{ and}\\
    %  \left(\beta^a_b \cdot m + 2 \alpha^i_{j} + 4,\; \alpha^a_b \cdot m + 2 \alpha^i_{j} + 4\right).
    %  \end{gathered}
  %\end{equation}

\subparagraph{Merge tile.}
Finally, we define a merge tile of $a$ and $W \subseteq [k]$ such that $a < b$ for every $b \in W$, denoted by \mergetile{$a$}{$W$}.
For every $i \in [n_a]$, it contains an atomic pair
\begin{equation*}
  Y^\msfm_{a,i} = \left\{
  \left(\alpha^a_i \cdot m + 3,\; \alpha^a_i \cdot m + 3\right), \; 
  \left((\alpha^a_i + 1)\cdot m + 2,\; \alpha^a_i \cdot m + 4\right)\right\}.
\end{equation*}

%For every $i \in [a-1]$ and every $j \in [n_i]$, it again contains an atomic pair $Y^{4a+2}_{i,j} = Y^{4a-1}_{i,j}$. Then it contains for every $j \in [n_a]$ the atomic pair $Y^{4a+2}_{a,j}$ consisting of points
%\begin{gather*}
  %  \left(\alpha^a_j \cdot m + 1,\; \alpha^a_j \cdot m + 1\right) \quad \text{ and } \quad
  %  \left((\alpha^a_j + 1)\cdot m,\; \alpha^a_j \cdot m + 2\right).
  %\end{gather*}

And for every $b \in W$ and $(v^a_i,v^b_{j}) \in E_a$, it contains an atomic pair
\begin{equation*}
  Z^\msfm_{a,i,b,j} = \left\{
  \begin{gathered}
    \left(\alpha^b_{j} \cdot m - \alpha^a_i + 3,\; \alpha^a_i \cdot m + 2 \alpha^b_{j} + 5 \right),\\
    \left((\alpha^b_{j} + 1) \cdot m  + \alpha^a_i+ 2 ,\; \alpha^a_i \cdot m + 2 \alpha^b_{j} + 6\right)
  \end{gathered}\right\}.
\end{equation*}
The merge tile is split into $n$ vertical strips, the first $n_a$ of them contains the atomic pairs $Y^\msfm_{a,i}$ and every following strip contains the pairs $Z^\msfm_{a,i,b,j}$ for fixed $b$, $j$ and all possible pairs of $a$ and $i$.
See the right part of Figure~\ref{fig:tiles2}.

Every tile other than \anchortile{} contains additionally two other atomic pairs called \emph{guards} defined as
\begin{equation*}
  G_1 = \left\{ (1,1), (2,2)\right\},\; \text{ and }\; G_2 = \left\{ (nm + 3, nm + 3), (nm+4 , nm+4)\right\}.
\end{equation*}
Observe that in any type of tile, $G_1$ lies to the left and below everything else while $G_2$ lies to the right and above everything else.

\subparagraph{Modified $\cF$-assembly.}
When creating a permutation from tiles consisting of atomic pairs, we slightly change each atomic pair as to force a specific relative order of two atomic pairs in neighboring tiles that share the same coordinates.
Suppose we are given a monotone gridding matrix $\cM$ whose cell graph is a tree.
Assume that the tree is rooted and oriented all the edges consistently outwards from the root.
For a vertex $v$ of the tree, its parent is the only in-neighbor of $v$.
Suppose we are given a consistent orientation $\cF$ of $G_\cM$ and a family of tiles $\cQ$ such that 
each tile is one of the types defined in this section.
We define a \emph{modified $\cF$-assembly} of $\cQ$ as follows.

Let $v$ be a vertex in $G_\cM$ with parent $w$.
Let $X = \{p,q\}$ be an atomic pair in $Q_v$ such that $p$ lies to the left and below of $q$.
If $w$ and $v$ share the same row, then we move $p$ up by a tiny distance (increasing its $y$-coordinate) and we move $q$ down by a tiny distance (decreasing its $y$-coordinate) without changing relative position of any points in tiles $T_v$ and $T_w$.
On the other hand, if $w$ and $v$ share the same column, we do the same modification with the 
$x$-coordinate, i.e., we increase the $x$-coordinate of $p$ and decrease the $x$-coordinate of~$q$.
Then we perform the usual $\cF$-assembly on this modified family of tiles.

We say that a pair $(a,b)$ of numbers \emph{sandwiches} the pair $(c,d)$ if $a < c < d < b$.
Notice that by this modification,
if there was an atomic pair $X$ in $Q_v$ and $Y$ in $Q_w$ then after the modification the 
$x$-coordinates of $Y$ sandwich the $x$-coordinates of $X$ if $v$ and $w$ share the same column.
Otherwise the same holds with respect to the $y$-coordinates of $X$ and $Y$.
This property can then be used to force a specific mapping of the points in $Q_v$ depending on the mapping of the points in $Q_w$.

\subsection{Reductions}
\label{ssec:reductions}

%TODO: add some intro to subsection

\begin{lemma}
  \label{lem:long-path-reduction}
  Let $\cC$ be a class with the computable LPP. An instance $(G,\chi)$ of \textsc{Partitioned Clique} can be reduced to an instance $(\pi, \tau, A, B)$ of \APPPM{$\cC$} where $|\pi| \in O(k^2)$ and $|\tau| \in O(|V_H|^2)$ in time $f(k) \cdot |V_H|^{O(1)}$ for some function $f$. Moreover, $\tw(\pi) \in O(k)$.
\end{lemma}

%\begin{proof}[Proof idea]
%Using the computable LPP, we obtain a monotone gridding matrix $\cM$ such that $\Grid(\cM)$ is a subclass of $\cC$ and the cell graph of $\cM$ is a proper-turning path with $4k-2$ vertices $v_1, \ldots, v_{4k-2}$.
%We construct the pattern $\pi$ via an $\cF$-assembly from a family of tiles $\cP$ and the text $\tau$ from a family of tiles $\cT$ where the only non-empty tiles in both families correspond to the non-empty entries of $\cM$ and moreover, each non-empty tile in $\cP$ is an increasing sequence.
%
%The first tiles $P_{v_1}$ and $T_{v_1}$ both contain only pair of elements whose images under the $\cF$-assembly are taken as the anchors $A$ and $B$.
%Their role is to guarantee that any embedding of $\pi$ into $\tau$ that maps $A$ to $B$ must be grid-preserving, i.e. it maps the image of the tile $P_{v_i}$ to the image of the tile $T_{v_i}$ for every $i$. The second pair of tiles $P_{v_2}$ and $T_{v_2}$ then simulates a mapping $\phi\colon [k] \to V_H$ that respects the coloring $\chi$.
%And finally for every $i \in [k]$, the tiles corresponding to vertices $v_{4a-1}$, $v_{4a}$, $v_{4a+1}$ and $v_{4a+2}$ verify that there is an edge in $H$ between the vertex $\phi(i)$ and $\phi(j)$ for every $j > i$.
%\end{proof}

\begin{proof}
  Due to the computable LPP property, we can obtain a monotone gridding matrix $\cM$ such that $\Grid(\cM) \subseteq \cC$, the cell graph $G_\cM$ is a proper turning path $v_1, \dots, v_{4k-2}$ of length $4k-2$ which we consider as an oriented path starting in $v_1$, and $v_1$ and $v_2$ share a common row.
  By Lemma~\ref{lem-acyclictrans}, there is a consistent orientation $\cF$ of $\cM$.
  We define two $k \times \ell$ families of tiles $\cP$ and $\cT$ which are then used for constructing $\pi$ and $\tau$ via a modified $\cF$-assembly.
  
  Given the input instance $(H, \chi)$ of \textsc{Partitioned Clique}, we define the tiles using the equivalent \PSI instance $(G, H, \chi)$ where $G$ is the clique on the vertex set $[k]$.
  We set $P_{v_1}$ and $T_{v_1}$ to be the anchor tiles. The tile $P_{v_2}$ is taken to be \patterntile{$[k]$} and $T_{v_2}$ is taken to be \assigntile{}.
  The rest of the tiles are defined in $k$ consecutive groups.
  Let $I_a = \{a, \dots, k\}$.
  For $a \in [k-1]$, we set $P_{v_i} = \patterntile{I_a}$ for every $4a-1 \le i \le 4a + 2$.
  On the text side, we set $T_{v_{4a - 1}} = \idtile{I_a}$, $T_{v_{4a}} = \branchtile{a}{I_{a+1}}$, $T_{v_{4a+1}} = \testtile{a}{I_{a+1}}$ and $T_{v_{4a+2}} = \mergetile{a}{I_{a+1}}$.
  
  Then we take $\pi$ as the modified $\cF$-assembly of $\cP$, $\tau$ as the modified $\cF$-assembly of $\cT$ and we set the anchors $A$, $B$ to be the images of $P_{v_1}$ and $T_{v_1}$.
  Clearly, $\pi \in \cC$ and thus, we constructed a valid instance $(\pi, \tau, A, B)$ of \APPPM{$\cC$} in polynomial time.
  Now we check that we achieved the desired upper bounds on size.
  The total size of the pattern tiles is clearly $O(k^2)$.
  The tiles $T_{v_{4a-1}}, \dots, T_{v_{4a+2}}$ each contain at most $2 n_a \cdot n$ points and therefore, the size of the text tiles sums to $O(n^2)$.
  
  %  However, we assumed that we are given the monotone matrix $\cM$.
  %  Luckily, we know that $\cM$ has exactly $4k-2$ non-empty entries and therefore, we can enumerate every $\cM$ with no empty rows or columns such that $G_\cM$ is a proper-turning path of length $4k-2$ in time $g(k)$ for some function $k$.
  %  Then we perform the modified $\cF$-assemblies described above to obtain permutations $\pi_\cM$ and $\tau_\cM$.
  %  Using the assumption that there is an algorithm to decide membership in $\cC$, we  can check in time $h(k)$ for some function $h$ whether $\pi_\cM$ belongs to $\cC$.
  %  Since $\cC$ has the long path property, we must find at least one suitable gridding matrix $\cM$ which gives us a valid instance $(\pi_\cM, \tau_\cM, A, B)$ for \APPPM{$\cC$} in time $g(k)\cdot h(k) \cdot n^{O(1)}$.
  
  \subparagraph{Correctness (``only if'').}
  Suppose that $(H, \chi)$ is a positive instance of \textsc{Partitioned Clique} and thus, 
$(G,H,\chi)$ is a positive instance of \PSIx.
  There is a witnessing mapping $\phi\colon [k] \to \mathbb{N}$ such that $\phi(a) \in [n_a]$ for 
every $a$ and $\{v^a_{\phi(a)}, v^b_{\phi(b)}\} \in E_H$ for every different $a, b \in [k]$.
  
  Let us define an embedding $\psi$ of $\pi$ into $\tau$ that maps $A$ to $B$.
  The embedding shall be \emph{grid-preserving}, meaning that the image of $P_{v_i}$ is mapped to the image of $T_{v_i}$.
  That way, we automatically guarantee that $A$ maps to $B$.
  First we map the guards in $P_{v_i}$ to the guards of $T_{v_i}$ for every $i \ge 2$.
  Then we define the mapping of points in $P_{v_2}$ to mimic the mapping $\phi$. For each $a \in [k]$, we set the image of the atomic pair $X_a$ in $P_{v_2}$ to be the atomic pair $Y^\msfa_{a,\phi(a)}$ in  $T_{v_2}$.
  If $P_{v_i} = \patterntile{W}$ and $T_{v_i} = \idtile{W}$, we map the atomic pair $X_a$ in $P_{v_i}$ to the atomic pair $Y^\msfid_{a,\phi(a)}$ in $T_{v_i}$.
  If $P_{v_i} = \patterntile{W}$ and $T_{v_i}$ is one of $\branchtile{a}{W}$, $\testtile{a}{W}$ or $\mergetile{a}{W}$ then we map $X_a$ in $P_{v_i}$ to the atomic pair $Y^\delta_{a,\phi(a)}$ in $T_{v_i}$ and for $b \in W$, we map $X_{b}$ to $Z^\delta_{a,\phi(a), b, \phi(b)}$ for the appropriate choice of $\delta \in \{\msfb, \msft, \msfm\}$.
  
  First, we need to make sure that the atomic pair $Z^\delta_{a,\phi(a),b, \phi(b)}$ for the appropriate choice of $\delta \in \{\msfb, \msft, \msfm\}$ is well-defined in all of $\branchtile{a}{W}$, $\testtile{a}{W}$ and $\mergetile{a}{W}$.
  Equivalently, we need to verify that $(v^a_{\phi(a)}, v^b_{\phi(b)}) \in E_a$ which is guaranteed by $\phi$.
  
  In order to check the validity of the embedding, it is sufficient to verify that the image of $P_{v_i}$ is an increasing sequence in $T_{v_i}$ and that furthermore, the images of $P_{v_i}$ and $P_{v_{i+1}}$ have the right relative order according to their $x$-coordinates if $P_{v_i}$ and $P_{v_{i+1}}$ share a common column ($i$ is even), an according to their $y$-coordinates if they share a common row ($i$ is odd).
  
  The increasing property is checked straightforwardly.
  Observe that for every $\delta \in \{\msfid, \msfa, \msfb, \msft, \msfm\}$ and $a, b \in [k]$ such that $a < b$, whenever the atomic pairs $Y^\delta_{a, i}$ and $Y^\delta_{b, j}$ in a single tile $T_{v_i}$ are defined, then $Y^\delta_{a, i}$ lies to the left and below of $Y^\delta_{b, j}$ regardless of the choices of $i \in [n_a]$ and $j \in [n_b]$.
  Moreover for a fixed $a$ and $b < b'$, the atomic pair $Z^\delta_{a,\phi(a),b,i}$ lies to the left and below $Z^\delta_{a,\phi(a),b',j}$ again regardless of the choices of $i$ and $j$.
  Finally, every defined atomic pair $Z^\delta_{a,\phi(a),b,j}$ lies to the right and above the pair $Y^\delta_{a,\phi(a)}$.
  
  In regards to the relative positions of points in different tiles, it is sufficient to verify that the image of $X_a$ in the tile $P_{v_i}$ under $\psi$ sandwiches the image of $X_a$ in the tile $P_{v_{i+1}}$.
  It is a mechanical task to check that from the definition of $\psi$ and the definitions of the respective tile types in subsection~\ref{subsec:building-blocks}.

  \subparagraph{Correctness (``if'').}
  Suppose there is an embedding $\psi$ of $\pi$ into $\tau$ that maps anchors $A$ to~$B$.
  The horizontal strip in $\tau$ between the anchors $B$ contains only the image of $T_{v_2}$ and the horizontal strip between the anchors $A$ in $\pi$ contains only the image of $P_{v_2}$.
  This implies that the guards in $P_{v_2}$ must map to the guards in $T_{v_2}$. 
  And inductively, the guards in $P_{v_i}$ sandwich the guards in $P_{v_{i+1}}$ and thus the guards in $P_{v_{i+1}}$ map to the only atomic pairs sandwiched in $T_{v_{i+1}}$ by the guards in $T_{v_i}$ which are the guards in $T_{v_{i+1}}$.
  These guards then force that the whole embedding is necessarily grid-preserving.
  
  Moreover, the atomic pair $X_a$ in $P_{v_2}$ for every $a \in [k]$ is mapped by $\psi$ to the atomic pair $Y^\msfa_{a,i}$ for some $i$.
  Let $\phi\colon [k] \to \mathbb{N}$ be the mapping such that $X_a$ is mapped precisely to $Y^\msfa_{a,\phi(a)}$ for every $a \in [k]$.
  Clearly, $\phi$ can be used to define a map that satisfies the coloring property -- $\chi(v^a_{\phi(a)}) = a$ for every $a \in [k]$.
  It remains to show that $\{v^a_{\phi(a)}, v^b_{\phi(b)}\} \in E_H$ for every pair of distinct $a, b \in [k]$.
  
  \begin{claim}
    Let $w = v_i$ for $i \ge 2$ and let $a \in [k]$. The atomic pair $X_a$ in $P_{w}$ (if defined) is mapped by $\psi$ to the atomic pair $Y^\delta_{a, \phi(a)}$ for appropriate $\delta$ if $T_w$ is one of $\idtile{U}$, $\branchtile{a}{W}$, $\testtile{a}{W}$ or $\mergetile{a}{W}$ where $a \in U$ and $a\not\in W$. For $b < a$, $X_a$ in $P_{w}$ is mapped to the pair $Z^\delta_{b,\phi(b), a,\phi(a)}$ if $T_w$ is one of $\branchtile{b}{U}$, $\testtile{b}{U}$ or $\mergetile{b}{U}$.
  \end{claim}
  \begin{proof}
    For $w = v_2$ the claim holds by the definition of $\phi$.
    We prove it inductively by the lexicographic order on $(a,i)$.
    
    Suppose that $T_{v_i} = \idtile{W}$ and recall that in fact $W = I_b$ for some $b \in [k]$.
    In such case the tile $T_{v_{i-1}}$ is either the tile $\assigntile$ or $\mergetile{b-1}{I_b}$.
    In the first case, we already know that $X_a$ in $P_{v_{i-1}}$ is mapped to $Y^\msfa_{a, \phi(a)}$ in $T_{v_{i-1}}$.
    Given the structure of these tiles, the atomic pair $X_a$ in $P_{v_i}$ is forced to map to the atomic pair $Y^\msfid_{a, \phi(a)}$ in $T_{v_i}$ as it is the only pair contained in the horizontal or vertical strip bounded by $Y^\msfa_{a, \phi(a)}$.
    In the second case, the same holds since $Y^\msfid_{a, \phi(a)}$ is the only pair in $T_{v_i}$ sandwiched by $Z^\msfm_{b, \phi(b),a,\phi(a)}$ from $T_{v_{i-1}}$.
    
    Suppose that $T_{v_i}$ is one of $\branchtile{a}{W}$, $\testtile{a}{W}$ or $\mergetile{a}{W}$ where $a \not \in W$.
    In these cases, we know that $X_a$ in $P_{v_{i-1}}$ maps to $Y^\delta_{a, \phi(a)}$ in $T_{v_{i-1}}$ for appropriate $\delta$.
    And again, this leaves only a single option where $X_a$ from $P_{v_{i}}$ can be mapped to -- the pair $Y^{\eta}_{a, \phi(a)}$ in $T_{v_i}$ for appropriate $\eta$.
    
    Suppose that $T_{v_i} = \branchtile{b}{W}$ for $b < a$ and notice that then $T_{v_{i+1}} = \testtile{b}{W}$.
    Using induction on $X_a$ in $P_{v_{i-1}}$ and the structure of $\branchtile{b}{W}$, we see that $X_a$ in $P_{v_i}$ can be mapped to $Z^\msfb_{b,j,a,\phi(a)}$ for any $j \in [n_b]$.
    However, this also forces the mapping of $X_a$ in $P_{v_{i+1}}$ to $Z^\msft_{b,j,a,\phi(a)}$ in $T_{v_{i+1}}$.
    Since $b < a$, we already know that $X_b$ in $P_{v_{i+1}}$ is mapped to $Y^\msft_{b, \phi(b)}$ in $T_{v_{i+1}}$ by the induction.
    And for $j \neq \phi(b)$, the pairs $Y^\msft_{b, \phi(b)}$ and  $Z^\msft_{b,j,a,\phi(a)}$ form an occurrence of $3412$ in $T_{v_{i+1}}$ which cannot happen as $P_{v_{i+1}}$ itself is an increasing sequence.
    Therefore, $X_a$ is mapped to $Z^\delta_{b,\phi(b), a,\phi(a)}$ in both $T_{v_i}$ and $T_{v_{i+1}}$ for appropriate $\delta$.
    In fact, this forces also the pair $X^a$ in $P_{v_{i+2}}$ to match to the pair $Z^\msfm_{b,\phi(b), a,\phi(a)}$ in $P_{v_{i+2}} = \mergetile{a}{W}$ which concludes the proof of the claim.
  \end{proof}

  Now let $a ,b \in [k]$ such that $a < b$.
  As we showed above, the atomic pair $X_b$  is mapped to the atomic pair $Z^\msft_{a,\phi(a), b, \phi(b)}$ in the tile $\testtile{a}{I_a}$.
  That means that $Z^\msft_{a,\phi(a), b, \phi(b)}$ is well-defined and thus, we have $\{v^a_{\phi(a)}, v^b_{\phi(b)}\} \in E_a$.
  Recall that we defined $G$ to be the clique on vertex set $[k]$.
  Therefore $\{a,b\} \in E_G$, and this necessarily implies that $\{v^a_{\phi(a)}, v^b_{\phi(b)}\} \in E_H$.
  Therefore, we see that $(G,H,\chi)$ is a positive instance of \PSI as witnessed by the map $\rho\colon V_G \to V_H$ obtained by setting $\rho(a) = v^a_{\phi(a)}$ and thus also $(H, \chi)$ is a positive instance of \textsc{Partitioned Clique}.
  
  \subparagraph{Tree-width of $\pi$.}Finally, let us show that $\tw(\pi) \in O(k)$.
  Let us say that an edge of $G_\pi$ is \emph{exceptional} if its endpoints share neither the same 
row nor the same column of the gridding.
  Only the lowest and highest point of each tile can participate in an exceptional edge.
  Therefore, there are at most $4k-2$ exceptional edges.
  Let $G'$ be the graph obtained from $G_\pi$ by removing all the exceptional edges.
  It is sufficient to show that $\tw(G')\in O(k)$ as adding the $4k-2$ edges back increases the 
tree-width at most by $4k-2$.
  
  We define a tree decomposition $(T, \gamma)$ such that $T$ is a path on $4k-1$ vertices $p_1, \dots, p_{4k-1}$ and $\gamma(p_i)$ is the image of  tiles $P_{v_i}$ and $P_{v_{i+1}}$.
  Clearly, every point of $\pi$ lies in a set of bags that induces a connected subtree.
  Moreover, every edge runs either inside a single cell or between points in two neighboring cells on the path since any other pair of points would occupy different row and different column.
  Therefore, we defined a valid tree decomposition of width $O(k)$ and $\tw(G') \in O(k)$. This 
completes the proof of Lemma~\ref{lem:long-path-reduction}.
\end{proof}

\begin{lemma}
  \label{lem:deep-tree-reduction}
  Let $\cC$ be a class with the computable DTP. An instance $(G,H, \chi)$ of \PSI can be reduced to an instance $(\pi, \tau, A, B)$ of \APPPM{$\cC$} where $|\pi| \in O(|E_G|\cdot \log|E_G|)$ and $|\tau| \in O(|E_H| + |V_H| \cdot |E_G|)$ in time $f(|E_G|) \cdot |V_H|^{O(1)}$ for some function $f$.
\end{lemma}

\begin{proof}
  In this reduction, we combine the ideas of the reduction for the LPP 
(Lemma~\ref{lem:long-path-reduction}) with the proof that DTP implies near-linear tree-width 
(Proposition~\ref{pro-tree}).
  
  Due to the computable DTP property, we can compute a monotone gridding matrix $\cM$ such that $\Grid(\cM) \subseteq \cC$ and its cell graph $G_\cM$ is a rooted tree $T$ with the following properties.
The root $r$ has a single child $r'$ sharing a common row, $T$ has exactly $k$ leaves, every  
non-root vertex with a single child is a corner, every leaf shares a common column with its neighbor 
and moreover, every leaf is at distance at least two from the nearest vertex with degree larger 
than~1.
  Moreover, the distance of any two vertices in $T$ is $O(\log k)$.
  The deep tree property guarantees an existence of such $\cM$, since we can always take a slightly larger tree and then cut off some branches to achieve the desired shape.
  
  Following along the proof of Proposition~\ref{pro-tree}, we orient the edges of $G_\cM$ consistently away from $r$ and for any vertex $v$, the descendants of $v$, denoted by $D(v)$, are all the out-neighbors of $v$.
  We arbitrarily order the edges $E_G = \{e_1, \dots, e_k\}$ and also the $k$ leaves of $G_\cM$ as 
$v_1, \ldots, v_k$, and we define the sets $A_w$ exactly as in~\eqref{eq:tree-Aw} on 
page~\pageref{eq:tree-Aw}.
  We additionally assume that $A_r = [k]$ which corresponds to $G$ having no isolated vertices.
  We again have $\sum_{v}A_v \in O(k \log k)$.
  
  We call the only neighbor of the leaf $v_i$ a \emph{petiole}, denoted it $u_i$, and we call the set of all vertices that are neither one of $r$, $r'$ nor the leaves or petioles the \emph{stem}.
  We define two families of tiles $\cP$ and $\cT$.
  We set both $P_{r}$ and  $T_{r}$ to be the anchor tile \anchortile{}.
  The tile $T_{r'}$ is set to be the assignment tile \assigntile{}.
  We set the tile $P_v$ for every vertex $v$ of the tree other than $r$ to the tile \patterntile{$A_v$}.
  For every $v$ that is part of the stem, we set $T_v$ to be the identity tile \idtile{$A_v$}.
  For $i \in [k]$ such that $e_i = \{a,b\}$ with $a < b$, we set $T_{u_i}$ to be the branch tile \branchtile{$a$}{$\{b\}$} and we set $T_{v_i}$ to be the test tile \testtile{$a$}{$\{b\}$}.
  
  By Lemma~\ref{lem-acyclictrans}, there is a consistent orientation $\cF$ of $\cM$.
  We set $\pi$ to be the modified $\cF$-assembly of $\cP$ and $\tau$ to be the modified $\cF$-assembly of $\cT$.
  And as before, we set $A$ as the image of $P_1$ and $N$ as the image of $T_1$.
  The total size of $\cP$ is clearly $O(k \log k)$.
  The size of all tiles in $\cT$ corresponding to leaves and petioles is $O(|E_H|)$ and every of the remaining $O(|E_G|)$ tiles contains at most $O(|V_H|)$ points which gives the desired bounds on the sizes of $\pi$ and $\tau$.
  
  %  Up to this point, we assumed that we are given the gridding matrix $\cM$.
  %  However, we can use the same trick as before and simply iterate over all possible such matrices since their number is bounded by $g(k)$ for some function $g$.
  %  For each $\cM$, we generate the corresponding pattern $\pi_\cM$ and text $\tau_\cM$ and then check in time $h(k)$ for some function $h$ if $\pi_\cM \in \cC$.
  %  We are guaranteed to find a valid instance $(\pi_\cM, \tau_\cM, A, B)$ of \APPPM{$\cC$} time $g(k)\cdot h(k) \cdot n^{O(1)}$.
  
  \subparagraph{Correctness.}
  The correctness essentially follows the same arguments as in the case of LPP.
  For the ``only if'' part, suppose that there is a mapping $\phi\colon [k] \to \mathbb{N}$ 
witnessing that $(G, H, \chi)$ is a positive instance of \PSIx.
  We define a grid-preserving mapping $\psi$ of $\pi$ into $\tau$.
  In the \assigntile{} tile, we map $X_a$ to $Y^\msfa_{a, \phi(a)}$ and in each identity tile $\idtile{W}$ where $a \in W$, we map $X_a$ to $Y^\msfid_{a, \phi(a)}$ .
  In the petiole $u_i$ and the leaf $v_i$ such that $e_i = \{a,b\}$ with $a < b$, the mapping $\psi$ sends $X_a$ to $Y^\delta_{a, \phi(a)}$ and $X_b$ to $Z^\delta_{a, \phi(a)b,\phi(b)}$ for suitable $\delta$.
  Observe that $Z^\delta_{a, \phi(a)b,\phi(b)}$ is well-defined since $\{v^a_{\phi(a)}, v^b_{\phi(b)} \} \in E_H$.
  The same arguments as in the case of LPP show that $\psi$ is indeed an anchored embedding of $\pi$ into $\tau$.
  
  For the ``if'' part, suppose there is an anchored embedding $\psi$ of $\pi$ into $\tau$.
  Following the guards from the root, we see that $\psi$ must be grid-preserving and as before, we define the mapping $\phi\colon [k] \to \mathbb{N}$ such that the pair $X_a$ in $P_{r'}$ maps to $Y^\msfa_{a, \phi(a)}$ in $T_{r'}$.
  It inductively follows that in any vertex of the stem, $X_a$ must be mapped to $Y^\msfid_{a,\phi(a)}$ (assuming $a \in W$).
  This holds in particular for the parent of each petiole $u_i$.
  Consequently, the petiole $u_i$ and leaf $v_i$ can be seen as applying the reduction from Lemma~\ref{lem:long-path-reduction} to check the subgraph property for the single edge $\{a,b\}$ in the subgraph of $H$ induced by $V_a \cup V_b$.
  Thus, the same arguments show that $\{v^a_{\phi(a)}, v^b_{\phi(b)}\} \in E_H$ and $(G, H, \chi)$ is 
a positive instance of \PSIx.
\end{proof}

Observe that both reductions produce $\pi$ and $\tau$ as gridded permutations belonging to some monotone grid class $\Grid(\cM)$ via an $\cF$-assembly from families of tiles.
Importantly, they share the property that any embedding of $\pi$ into $\tau$ that maps $A$ to $B$ 
must be grid-preserving, i.e., it maps the $(i,j)$-cell of the gridding of $\pi$ to the $(i,j)$-cell 
of the gridding of $\tau$ for every $i$ and $j$.
Moreover, both $A$ and $B$ are pairs of consecutive points in the left-to-right order.

\subsection{Consequences}

\begin{theorem}
\label{thm:tw-bound}
If $\cC$ has the computable LPP then \PPPM{$\cC$} cannot be solved in time $f(t) \cdot n^{o(t)}$ 
where $t = \tw(\pi)$ for any function $f$, unless ETH fails.
\end{theorem}

\begin{proof}
  Let $(\pi,\tau, A,B)$ be the instance of \APPPM{$C$} produced by Lemma~\ref{lem:long-path-reduction} and let $m$ be the length of $\tau$.
  We define $\pi'$ as the permutation obtained from $\pi$ by inflating both of the anchors in $A$ with  either two increasing or decreasing sequences of length $m$ such that $\pi'$ is still contained in $\cC$.
  Recall that one of these inflations is always possible.
  And similarly, we let $\tau'$ be the permutation obtained from $\tau$ by inflating both of the anchors in $B$ with the same type of monotone sequences of length $m$ as in $\pi'$.
  
  We claim that $\pi'$ is contained in $\tau'$ if and only if $(\pi,\tau, A,B)$ is a positive instance of \APPPM{$C$}.
  It is clear that if there is an embedding of $\pi$ into $\tau$ that maps $A$ to $B$, then there is an embedding of $\pi'$ into $\tau'$.
  
  For the other direction, assume there is an embedding $\phi$ of $\pi'$ into $\tau'$.
  The inflated anchors in $\pi'$ contain exactly $2m$ points while $\tau'$ contains only $m-2$ points outside of its inflated anchors.
  Therefore, at least $m+2$ points of the inflated anchors in $\pi'$ are mapped by $\phi$ to the inflated anchors in $\tau'$ and in particular, there must be at least one point in each of the  anchors in $\pi'$ mapped to the corresponding anchor in $\tau'$.
  Since the anchors $A$ and $B$ are pairs of consecutive points, observe that we can, in fact, map the whole inflated anchors in $\pi'$ to the inflated anchors in $\tau'$.
  It follows that we obtain a desired anchored embedding of $\pi$ into $\tau$ by deflating the anchors back to a single point.
  
  Finally, we show that $\tw(\pi') \le \tw(\pi) + 2$.
  The desired bound follows as otherwise, we could use a faster algorithm for \PPPM{$\cC$} to decide the instance $(\pi,\tau, A,B)$ of \APPPM{$C$} and consequently refute ETH by the ``moreover'' part of Lemma~\ref{lem:long-path-reduction}.
  We claim that in general, if $\sigma'$ is obtained from $\sigma$ by inflating one point with a monotone sequence then  $\tw(\sigma') \le \tw(\sigma) + 1$.
  To see that, notice that when we inflate a point of $\sigma$ with a monotone sequence of length 2, we get $\sigma'$ such that  $\tw(\sigma') \le \tw(\sigma) + 1$.
  However, if we inflate the same point by a longer monotone sequence and obtain a permutation $\sigma''$ then $G_{\sigma''}$ can be obtained by edge subdivisions from $G_{\sigma'}$, and it is well-known that subdividing en edge does not increase tree-width.
\end{proof}

In order to show the hardness of \PSPPM{$\cC$}, we first reduce to an intermediate problem called \emph{$\cC$-Pattern Surjective Colored PPM} (\SCPPPM{$\cC$}) whose input consists of a pattern $\pi \in \cC$, a text $\tau$ and a coloring $\chi\colon \tau \to [t]$. We need to decide whether there is an embedding of $\pi$ into $\tau$ that hits all $t$ possible colors.
This intermediate reduction allows us to infer conditional lower bounds for \PSPPM{$\cC$} via the following lemma.

\begin{lemma}[Berendsohn~\cite{BerendsohnMs}]
  \label{lem:berendsohn-reduction}
  Let there be an algorithm that solves \PSPPM{$\cC$} in time $f(k) \cdot n^{O(g(k))}$ for some functions $f$ and $g$. Then \SCPPPM{$\cC$}  can be solved in time $h(k) \cdot n^{O(g(k))}$ for some function $h$.
\end{lemma}

\begin{lemma}
  \label{lem:anch-surj-reduction}
  An instance $(\pi, \tau, A, B)$ of \APPPM{$\cC$} produced by Lemma \ref{lem:long-path-reduction} or \ref{lem:deep-tree-reduction} can be reduced to an instance $(\pi', \tau', \chi)$ of \SCPPPM{$\cC$} where $|\pi'| \in O(|\pi|)$ and $|\tau'| \in O(|\tau|)$ in polynomial time.
\end{lemma}
\begin{proof}
  The general idea of the proof is the same as in Theorem~\ref{thm:tw-bound} -- we force matching of the anchors by inflating them with long monotone sequences.
  The \SCPPPM{$\cC$} problem, however, allows us to use sequences with length depending only on $\pi$. 
  Let $k$ be the length of $\pi$ and let $\pi'$ be the permutation obtained by inflating the anchors $A$ with either two increasing or decreasing sequences of length $k$ such that $\pi' \in \cC$, and let $\tau'$ be the permutation obtained by the same inflation of the anchors $B$.
  We define $\chi\colon \tau \to [2k+1]$ by coloring every point added during the inflation with a unique color and using a single additional color for every other point.
  Clearly, $|\pi'| \in O(|\pi|)$ and $|\tau'| \in O(|\tau|)$.

  We need to verify the correctness of our construction. If $(\pi, \tau, A, B)$ is a positive instance of \APPPM{$\cC$} then $(\pi', \tau', \chi)$ is a positive instance of \SCPPPM{$\cC$} as we can simply map the inflated anchors of $\pi'$ to the inflated anchors of $\tau'$.
  For the other direction, assume there is an embedding $\phi$ of $\pi'$ into $\tau'$ that hits all the $2k+1$ colors.
  In other words, the image of $\pi'$ under $\phi$ contains the whole inflated anchors of $\tau'$.
  Since there are only $k-2$ points in $\pi'$ outside of the anchors, at least $k+2$ points of the anchors in $\pi'$ maps to the anchors in $\tau'$.
  In particular, there must be at least one point in each of the two increasing inflated anchors in $\pi'$ that maps to the corresponding anchor in $\tau'$.
  By the same argument as in the proof of Theorem~\ref{thm:tw-bound}, we conclude that the inflated anchors map without loss of generality to the inflated anchors.
%  And thus, we obtain the desired anchored embedding of $\pi$ into $\tau$ by %replacing the anchors back with single points.
\end{proof}

\begin{theorem}\label{thm:main}
Unless ETH fails, \PSPPM{$\cC$} cannot be solved for any function $f$
\begin{itemize}
  \item in time $f(k) \cdot n^{o\left(\sqrt{k}\right)}$ if $\cC$ has the computable LPP, and
  \item in time $f(k) \cdot n^{o\left(k/\log^2 k\right)}$ if $\cC$ has the computable DTP.
\end{itemize}
\end{theorem}

\begin{proof}
  For $\cC$ with the computable LPP, a faster algorithm would refute ETH via
  \[ \small \text{\textsc{Partitioned Clique}}
  \oto{Lemma~\ref{lem:long-path-reduction}}
  \text{\APPPM{$\cC$}} \oto{Lemma~\ref{lem:anch-surj-reduction}}
  \SCPPPM{$\cC$} \oto{Lemma~\ref{lem:berendsohn-reduction}} \PSPPM{$\cC$}\]
  
  Whereas for $\cC$ with the computable DTP, a faster algorithm would refute ETH via
    \[ \small \text{\PSI}
  \oto{Lemma~\ref{lem:deep-tree-reduction}}
  \text{\APPPM{$\cC$}} \oto{Lemma~\ref{lem:anch-surj-reduction}}
  \SCPPPM{$\cC$} \oto{Lemma~\ref{lem:berendsohn-reduction}} \PSPPM{$\cC$}. \qedhere\]
\end{proof}

\bibliography{bibliography}

\end{document}